\documentclass[12pt]{article}
\usepackage{amsmath,amsbsy}
\usepackage{graphicx,psfrag,epsf}
\usepackage{enumerate}
\usepackage[round]{natbib}
\usepackage{amsthm}
\usepackage{amsfonts}
\usepackage{amssymb}
\usepackage{bm}
\usepackage{xcolor}
\usepackage{algorithm}
\usepackage{algorithmic}
\usepackage{natbib}
\usepackage{graphicx}
\usepackage{amsthm}
\usepackage{cite}
\usepackage{graphicx}
\usepackage{textcomp}
\usepackage[section]{placeins}
\usepackage{setspace}
\usepackage{times}
\usepackage{comment}
\usepackage{multirow}
\usepackage{booktabs}
\usepackage{siunitx}
\usepackage{accents}
\usepackage{enumitem}
\usepackage{subcaption}
\usepackage{url} % not crucial - just used below for the URL 
\usepackage{float}
\usepackage{hyperref}

\def\Hcal{{\mathcal{H}}}

\def\Ncal{{\mathcal{N}}}

\def\Scal{{\mathcal{S}}}

\def\Ycal{{\mathcal{Y}}}

\def\Ibf{{\mathbf{I}}}

\def\Ebb{{\mathbb{E}}}

\def\Rbb{{\mathbb{R}}}
\def\Sbb{{\mathbb{S}}}

\newtheorem{definition}{Definition}
\newtheorem{theorem}{Theorem}
\newcommand\independentt{\protect\mathpalette{\protect\independenT}{\perp}}
\def\independenT#1#2{\mathrel{\rlap{$#1#2$}\mkern2mu{#1#2}}}

\def\Expect{\mathrm{E}}

\def\iter{\mathrm{iter}}

\def\calF{\mathcal{F}}
\def\calY{\mathcal{Y}}

\def\bC{{\bf C}}

\def\bQ{{\bf Q}}
\def\bU{{\bf U}}
\def\bV{{\bf V}}
\def\bZ{{\bf Z}}

\def\bbeta{{\boldsymbol{\beta}}}

\def\bC{{\bf C}}

\def\bX{{\bf X}}

\def\bK{{\bf K}}
\def\bM{{\bf M}}
\def\bY{{\bf Y}}

\def\bZ{{\bf Z}}

\def\calH{{\mathcal H}}
\def\calS{{\mathcal S}}

\newcommand{\indep}{\perp \!\!\! \perp}

\newcommand{\blind}{0}

\begin{document}

\def\spacingset#1{\renewcommand{\baselinestretch}%
{#1}\small\normalsize} \spacingset{1}

%%%%%%%%%%%%%%%%%%%%%%%%%%%%%%%%%%%%%%%%%%%%%%%%%%%%%%%%%%%%%%%%%%%%%%%%%%%%%%

\if0\blind
{
  \title{\bf Fr\'echet Sufficient Dimension Reduction for Metric Space-Valued Data via Distance Covariance}
  \author{Hsin-Hsiung Huang\\
    Department of Statistics and Data Science, University of Central Florida, Orlando, Florida\\
    and \\
    Feng Yu \\
    Department of Mathematical Science, University of Texas at El Paso, El Paso, Texas \\
    and \\
    Kang Li and Teng Zhang\thanks{
    CONTACT: Teng Zhang teng.zhang@ucf.edu Department of Mathematics, University of Central Florida, Orlando, Florida 32816, U.S.A.}\hspace{.2cm}\\
 Department of Mathematics, University of Central Florida, Orlando, Florida}
  \maketitle
} \fi

\if1\blind
{
  \bigskip
  \bigskip
  \bigskip
  \begin{center}
    {\LARGE\bf  Fr\'echet Sufficient Dimension Reduction for Metric Space-Valued Data via Distance Covariance}
\end{center}
  \medskip
} \fi

\bigskip
\begin{abstract}

We propose a novel Fr\'echet sufficient dimension reduction (SDR) method based on kernel distance covariance, tailored for metric space-valued responses such as count data, probability densities, and other complex structures.
The method leverages a kernel-based transformation to map metric space-valued responses into a feature space, enabling efficient dimension reduction. By incorporating kernel distance covariance, the proposed approach offers enhanced flexibility and adaptability for datasets with diverse and non-Euclidean characteristics. The effectiveness of the method is demonstrated through synthetic simulations and several real-world applications. In all cases, the proposed method runs faster and consistently outperforms the existing Fr\'echet SDR approaches, demonstrating its broad applicability and robustness in addressing complex data challenges.
\end{abstract}

\noindent%
\vfill

\newpage
\spacingset{1.45} %

\section{Introduction}

In contemporary data analysis and visualization, the emergence of metric-space data has become increasingly prevalent in applications involving complex, non-Euclidean data objects. These objects, often termed random objects, encompass diverse forms such as count responses, probability density functions (e.g., global mortality distributions and bike rental frequencies), carcinoma gene expression profiles, and breast cancer survival statistics. The inherent complexity and non-Euclidean geometry of such data challenge the foundational assumptions of conventional statistical methods, rendering them insufficient for precise and effective analysis. Addressing these challenges necessitates the development of innovative approaches that account for the unique characteristics of metric-space data.

The concept of the Fr\'echet mean was introduced in 1948 \citep{Frechet1948elements}, which extends the notion of the expectation of a random vector to random objects in a general metric space. Building on this foundational concept, Petersen and M\"uller proposed the Fr\'echet regression model \citep{petersen2019Frechet}, which generalizes the Fr\'echet mean to conditional Fr\'echet means, offering a versatile framework for regression analysis with responses residing in a metric space.

Despite these advancements, the application of Fr\'echet regression faces the significant challenge of the curse of dimensionality, particularly when dealing with high-dimensional predictor spaces. As the number of predictors increases, the efficiency and accuracy of regression models tend to diminish, necessitating dimension reduction techniques to preserve the essential information about the predictors and their relationships with the response variable. Sufficient Dimension Reduction (SDR) techniques \citep{huang2024robust,zhang2023dimension}, widely employed in classical regression contexts, offer a promising solution by projecting high-dimensional predictors onto a lower-dimensional subspace without losing essential information regarding the response. SDR, which is based on the conditional distribution of the response \citep{cook1991discussion,Li1991,xia2002adaptive,Yin11-AOS950}, achieves this by ensuring that the reduced predictors retain all the relevant regression information. Recently, Fr\'echet SDR methods \citep{ying2022frechet,zhang2023dimension,weng2023sparse} have been proposed to map metric-space-valued responses to real-valued random variables using a class of functions, followed by the application of classical SDR techniques to the transformed data. These Fr\'echet SDR methods operate within the classical framework, where the sample size $ n $ exceeds the dimension of the predictors $ p $. However, these methods encounter significant challenges in high-dimensional settings where $ p $ is much larger than $ n $, particularly in relation to matrix inversion and fitting inverse regression models. First, these methods require the inversion of a $ p \times p $ covariance matrix, which becomes problematic when $ p > n $ \citep{li2008sliced}. Second, interpreting the extracted components proves difficult due to the nonlinear relationships between features and responses \citep{tan2018convex,ying2022frechet}.

Recent developments in SDR methods using distance covariance (dCov) \citep{sheng2013direction,sheng2016sufficient,yang2024functional} have shown that these techniques can overcome the limitations of requiring constant covariance conditions or specific distributional assumptions for $X$, $X \mid Y$, or $Y \mid X$, making them broadly applicable to continuous and discrete variables across diverse distributions. Moreover, several robust SDR methods have been proposed to enhance coefficient estimation, including robust SDR using ball covariance \citep{zhang2019robust} and expected likelihood-based methods that minimize the Kullback-Leibler distance \citep{yin2005direction,zhang2015direction}.

Comparing with the functional SDR method proposed by \citep{yang2024functional}, our approach introduces significant advancements by expanding the applicability of sufficient dimension reduction to metric-space-valued responses, providing enhanced flexibility beyond purely scalar outcomes. Specifically, our method employs kernel-based feature mapping to accommodate responses in non-Euclidean spaces, utilizing advanced kernel properties such as the Central Mean Subspace (CMS) ensembles to establish robust theoretical guarantees. Furthermore, our approach effectively addresses the curse of dimensionality through a screening method \citep{pan2019generic} that avoids the computational burden of covariance inversion, instead optimizing over Stiefel manifolds. This innovation ensures suitability for ultra-high-dimensional scenarios. While the Fr\'echet SDR methods \citep{zhang2023dimension, weng2023sparse} are primarily tailored for functional data analysis, our method extends its applicability to a broader spectrum of response types, including multivariate and distributional data, thereby offering a more versatile framework.

In this paper, we present a novel sufficient dimension reduction (SDR) approach within the framework of Fr\'echet regression. Our method, Fd-SDR, leverages distance covariance and adopts a non-parametric strategy by mapping random objects in a metric space to real-valued random variables through a carefully selected class of functions. Compared to existing Fr\'echet SDR methods \citep{zhang2023dimension, weng2023sparse}, Fd-SDR demonstrates superior performance under general settings and is computationally faster. We detail our proposed method in Section~\ref{sec:meth}, providing a comprehensive theoretical analysis in Section~\ref{sec:theory}, including justification on kernel as well as the convergence and consistency properties. Numerical experiments validating the effectiveness of Fd-SDR on both synthetic and real datasets are presented in Sections~\ref{sec:syn} and \ref{sec:real_datasets}, respectively.

\section{Methodology}\label{sec:meth}
This section presents our methodology of robust regularized Fr\'echet regression for sufficient dimension reduction. Section~\ref{sec:robustsdr_alpha} reviews the $\alpha$-dCov SDR model, and in Section~\ref{sec:gen_fsdr}, we propose our method of Fr\'echet SDR model via kernel distance covariance (KdCov).

\subsection{Sufficient Dimension Reduction via dCov}\label{sec:robustsdr_alpha}
This section reviews the classic SDR approach using the $\alpha$-dCov model, in which the responses are either scalars or vectors. The $\alpha$-dCov-based SDR framework generalizes the dCov method by improving robustness against data corruption caused by outliers \citep{huang2024robust}. Notably, the $\alpha$-dCov-based framework is reduced to the classical dCov-based SDR \citep{sheng2013direction,sheng2016sufficient} when $\alpha = 1$, and the term ``dCov-based SDR" specifically refers to this case throughout this paper.

Let $\{(X_i, Y_i)\}_{i=1}^n$ be $n$ i.i.d. samples drawn from the random variables $(X, Y)$. We define the data matrix $\bX = [X_1,\ldots,X_n] \in \Rbb^{p \times n}$, where the columns are $X_1, \ldots, X_n$, and let $\bY = [Y_1, \ldots, Y_n] \in \Rbb^n$ represent the vector of scalar responses. Let $\bm{\beta}\in\Rbb^{p\times d}$ be the coefficient matrix and $\hat{\Sigma}_X\in\Rbb^{p\times p}$ be the covariance matrix of $\bX$, which is assumed to be nonsingular. The coefficient $\bm{\beta}$ of the $\alpha$-dCov SDR model is obtained via maximizing the $\alpha$-dCov of $\bm{\beta}^T\bX$ and $\bY$, $\nu^{2}(\bm{\beta}^T\bX,\bY;\alpha)$, and its empirical solution is given by solving the following optimization problem \citep{huang2024robust,sheng2016sufficient}:
\begin{align}\label{eq:empirical_alpha_dcov}
\max_{\bm{\beta}\in\mathbb{R}^{p\times d}}\nu^2_n(\bm{\beta}^T \bX,\bY;\alpha), \mbox{ s.t. }\bm{\beta}^T\hat{\Sigma}_X\bm{\beta}=\Ibf_d,
\end{align}
where $\nu^2_n(\cdot,\cdot;\alpha)$ is the empirical $\alpha$-dCov with $0<\alpha<2$~\citep{10.1214/009053607000000505}. Here, the $\alpha$-dCov is defined as follows: 
\begin{align}\label{eq:nu}
\nu^{2}(X,Y;\alpha )&:=\Ebb [\|X-X'\|^{\alpha }\,\|Y-Y'\|^{\alpha }]+\Ebb [\|X-X'\|^{\alpha }]\,\Ebb [\|Y-Y'\|^{\alpha }]\nonumber\\
&-2\Ebb [\|X-X'\|^{\alpha }\,\|Y-Y''\|^{\alpha }].
\end{align}
and the empirical $\alpha$-dCov is defined as follows.
We define the distance matrices $(a_{kl}) = (\|X_k - X_l \|^{\alpha}_2), (b_{kl}) = (|Y_k -Y_l |^{\alpha})$ for $0<\alpha <2$ and $\forall k,l=1,\ldots,n$. The matrix $A\in\Rbb^{n\times n}$ is defined by 
\begin{align*}
    A_{kl} = a_{kl} - \bar{a}_{k\cdot} - \bar{a}_{\cdot l} + \bar{a}_{\cdot\cdot},\quad \forall k,l=1,\ldots,n,
\end{align*}
where the row and column means, along with the overall mean, are given by 
\begin{align*}
    \bar{a}_{k\cdot} = \frac{1}{n}\sum^n_{l=1}a_{kl},\quad \bar{a}_{\cdot l},= \frac{1}{n}\sum^n_{k=1}a_{kl},\quad \bar{a}_{\cdot\cdot} = \frac{1}{n^2}\sum^n_{k,l=1}a_{kl}.
\end{align*}
Similarly, the matrix $B\in\Rbb^{n\times n}$ is defined analogously 
\begin{align}\label{eq:b_dd}
    B_{kl} = b_{kl} - \bar{b}_{k\cdot} - \bar{b}_{\cdot l} + \bar{b}_{\cdot\cdot},\forall k,l=1,\ldots,n.
\end{align}
Using these matrices, the empirical $\alpha$-dCov, $\nu_n^2(\bX,\bY)$, is defined by~\citep[Definition 4]{10.1214/009053607000000505}
\begin{align}\label{eq:alt_dcov}
    \nu_n^2(\bX,\bY;\alpha) = \frac{1}{n^2}A_{kl}B_{kl}.
\end{align}
A computational advantage of the $\alpha$-dCov approach is that the problem in \eqref{eq:empirical_alpha_dcov} can be solved efficiently. Indeed, there is an alternative formulation of $\nu_n^2(\bX,\bY)$~\citep[Appendix]{10.1214/009053607000000505} as follows:
\begin{align*}
    \nu_n^2(\bX,\bY;\alpha) = S_1+S_2-2S_3,
\end{align*}
where 
\begin{align*}
    S_1 & = \frac{1}{n^2}\sum_{k,l=1}^na_{kl}b_{kl}, \\
    S_2 & = \frac{1}{n^2}\sum_{k,l=1}^na_{kl} \frac{1}{n^2}\sum_{k,l=1}^nb_{kl}=\frac{1}{n^2}\sum_{k,l=1}^na_{kl}\bar{b}_{\cdot\cdot}, \\
    S_3 & = \frac{1}{n^3}\sum_{k=1}^n\sum_{l,m=1}^na_{kl}b_{km} = \frac{1}{n^2}\sum_{k,l=1}^na_{kl}\bar{b}_{k\cdot}.
\end{align*}
Notice that $\frac{1}{n^2}\sum_{k,l=1}^na_{kl}\bar{b}_{k\cdot}=\frac{1}{n^2}\sum_{k,l=1}^na_{kl}\bar{b}_{\cdot l}$ because for any $k,l\in[n]$, $a_{kl}\bar{b}_{k\cdot}=a_{lk}\bar{b}_{\cdot k}$. So we have that 
\begin{align}\label{eq:alt_emp_dcov}
    \nu_n^2(\bX,\bY;\alpha) = S_1+S_2-2S_3 = \frac{1}{n^2}\sum_{k,l=1}^na_{kl}(b_{kl}+\bar{b}_{\cdot\cdot}-\bar{b}_{k\cdot}-\bar{b}_{\cdot l}) = \frac{1}{n^2}a_{kl}B_{kl}
\end{align}

Consider the transformation $\bC = \hat{\Sigma}^{\frac12}_X \bm{\beta}$ and $\bZ = \hat{\Sigma}^{-\frac12}_X \bX$. With the alternative formulation of the empirical dCov in \eqref{eq:alt_emp_dcov}, we can reformulate the problem in \eqref{eq:empirical_alpha_dcov} as the following optimization problem:
\begin{align}\label{eq:alpha_dcov}
    \max_{\bC} \nu^2_n (\bC^T \bZ, \bY,\alpha) := \frac{1}{n^2}\sum^n_{k,l=1}a_{kl}(\bC;\alpha)B_{kl},\mbox{ s.t. }\bC \in \mbox{St}(d, p),
\end{align}
where $a_{kl}(\bC;\alpha)=\|\bC^T Z_{k}-\bC^T Z_{l}\|^\alpha_2$ and $B_{kl}$ is defined in \eqref{eq:b_dd}. Here the constraint $\mbox{St}(d, p) = \{\bC \in \mathbb{R}^{p\times d}\mid \bC^T\bC = I_d\}$ with $d \leq p$ denotes the Stiefel manifold.

\subsection{Fr\'echet Sufficient Dimension Reduction via KdCov}\label{sec:gen_fsdr}

The classic dCov-based SDR model is primarily designed for vector-valued responses. To accommodate metric space-valued responses, we replace the standard distance covariance with KdCov, as defined in Definition~\ref{def:kernel_dCov}, and provided its empirical version in Definition~\ref{def:empirical_kernel_dCov}. This concept, originally introduced by \citep{sejdinovic2013equivalence}, differs from the standard distance covariance given in \eqref{eq:nu}. Notably, the KdCov in \citep{sejdinovic2013equivalence} is formulated for both $X$ and $Y$ defining in metric spaces. In this paper, we focus on predictors $X$ in Euclidean space, making the KdCov defined in Definition~\ref{def:kernel_dCov} a specialized case of the general framework proposed in \citep{sejdinovic2013equivalence}.

\begin{definition}[KdCov~\citep{sejdinovic2013equivalence}]\label{def:kernel_dCov}
    Let $X\in\Rbb^p$ be a random vector and $Y\in\Ycal$ be a metric space-valued random variable, where $\Ycal$ is a metric space. 
    Let $(X', Y'),(X'',Y'')$ be two i.i.d. copies drawn from the joint distribution of $(X,Y)$. Assume there exists a reproducing kernel Hilbert space (RKHS) $\calH_Y$ with an associated feature mapping $\phi: \calY\rightarrow \calF$, where $\calF$ is a Hilbert space referred to as feature space. The kernel distance covariance (KdCov) is defined as follows:
    \begin{align}\label{eq:kernal_dCov}
    \nu^2_{\Hcal}(X,Y) &:= \mathbb{E}\left[\|X - X'\| d(Y,Y')\right] + \mathbb{E}\left[\|X - X'\|\right]\mathbb{E}\left[d(Y,Y')\right] - 2\mathbb{E}\left[\|X - X'\| d(Y,Y'')\right],
    \end{align}
    where $d(Y,Y')$ is the distance in the feature space and is given by
    \begin{align}\label{eq:rkhs1}
    d(Y,Y') &= \|\phi(Y) - \phi(Y')\| = \sqrt{\langle \phi(Y) - \phi(Y'), \phi(Y) - \phi(Y') \rangle} \nonumber\\
    &= \sqrt{\kappa(Y,Y) + \kappa(Y',Y') - 2\kappa(Y,Y')}.
    \end{align}
    Here $\kappa(\cdot,\cdot):\Ycal\times\Ycal\rightarrow\Rbb$ is the kernel function, and defined by $\kappa(Y,Y')=\langle\phi(Y),\phi(Y')\rangle$.
\end{definition}

\begin{definition}[empirical KdCov]\label{def:empirical_kernel_dCov}
For a random sample of $(\bX,\bY)=\{(X_i,Y_i)\}_{i=1}^n$ from the joint distribution of random vector $X\in\Rbb^{p}$ and random variable $Y\in\Ycal$ where $\Ycal$ is a metric space. Assume there exists a reproducing kernel Hilbert space (RKHS) $\calH_Y$ with an associated feature mapping $\phi: \calY\rightarrow \calF$. Define $a_{kl}=\|X_k-X_l\|_2,b_{kl}=\|\phi(Y_k)-\phi(Y_l)\|_2$ for $k,l\in[n]$ and the row and column means, along with the overall mean, of $a_{kl}$ and $b_{kl}$ are further defined by 
\begin{align}
    \bar{a}_{k\cdot} = \frac{1}{n}\sum^n_{l=1}a_{kl},\quad \bar{a}_{\cdot l},= \frac{1}{n}\sum^n_{k=1}a_{kl},\quad \bar{a}_{\cdot\cdot} = \frac{1}{n^2}\sum^n_{k,l=1}a_{kl} \label{eq:ekdcov_a}\\
    \bar{b}_{k\cdot} = \frac{1}{n}\sum^n_{l=1}b_{kl},\quad \bar{b}_{\cdot l},= \frac{1}{n}\sum^n_{k=1}b_{kl},\quad \bar{b}_{\cdot\cdot} = \frac{1}{n^2}\sum^n_{k,l=1}b_{kl}. \label{eq:ekdcov_b}
\end{align}
The empirical kernel distance covariance is the nonnegative number defined by
\begin{align*}
\nu^2_{\Hcal,n}(\bX,\bY) &:= \frac{1}{n^2}A_{kl}B_{kl},
\end{align*}
where 
\begin{align*}
    A_{kl} = a_{kl} - \bar{a}_{k\cdot} - \bar{a}_{\cdot l} + \bar{a}_{\cdot\cdot}, \quad B_{kl} = b_{kl} - \bar{b}_{k\cdot} - \bar{b}_{\cdot l} + \bar{b}_{\cdot\cdot}, \quad \forall k,l=1,\ldots,n.
\end{align*}

\end{definition}

KdCov effectively characterizes the dependency structure between the predictors $\bX$ and the responses $\bY$ in a metric space. Notably, when $\phi$ is the identity mapping (which is associated with the linear kernel), KdCov is reduced to the standard distance covariance defined in \eqref{eq:nu}. This argument also implies that such distance can includes the Wasserstein metric between distributions and Frobenius norm between matrices as special examples, since that the Wasserstein metric between distributions $\mu_1$ and $\mu_2$ can be considered as the $\ell_2$ distance between $F_1^{-1}$ and $F_2^{-1}$ over $[0,1]$, where $F_i$ and the cumulative distribution functions of $\mu_i$, and the Frobenius norm between matrices can be considered as the Euclidean distance between their vectorized forms. By utilizing KdCov, the proposed SDR method is naturally extended to non-Euclidean settings, providing a more general framework for analyzing complex data structures.

To find the linear dependence between $\bX$ and $\bY$, we propose to maximize the KdCov of $\bX$ and $\bY$ as follows:
\begin{align}\label{eq:max_kdcov}
\max_{\bm{\beta}\in\mathbb{R}^{p\times d}}\nu^2_{\Hcal}(\bm{\beta}^T \bX,\bY), \mbox{ s.t. }\bm{\beta}^T\hat{\Sigma}_X\bm{\beta}=\Ibf_d,
\end{align}
where $\bm{\beta}\in\mathbb{R}^{p\times d}$ is the coefficient matrix and $\hat{\Sigma}_X\in\Rbb^{p\times p}$ is the covariance matrix of $\bX$. To solve this problem, we consider the transformation $\bC = \hat{\Sigma}^{\frac12}_X \bm{\beta}$ and $\bZ = \hat{\Sigma}^{-\frac12}_X \bX$, leading to the constrained maximization problem:
\begin{align}\label{eq:empirical_kdcov}
    \max_{\bC\in\Rbb^{p\times d}} \nu^2_{\Hcal,n} (\bC^T \bZ, \bY) := F(\bC),\mbox{ s.t. }\bC \in \mbox{St}(d, p),
\end{align}
where $\mbox{St}(d, p) = \{\bC \in \mathbb{R}^{p\times d}\mid \bC^T\bC = I_d\}$ with $d \leq p$ is the Stiefel manifold and $\nu^2_{\Hcal,n}(\cdot,\cdot)$ represents the empirical KdCov provided in Definition~\ref{def:empirical_kernel_dCov}. Similar to the empirical dCov, the empirical KdCov also admits an alternative formulation. Following the formulas in Definition~\ref{def:empirical_kernel_dCov}, the objective $F(\bC)$ in \eqref{eq:empirical_kdcov} is given by
\begin{align}\label{eq:empirical_kdcov_F}
    F(\bC) := \frac{1}{n^2}\sum^n_{k,l=1}A_{kl}(\bC)B_{kl},
\end{align}
where $A_{kl}(\bC)$ and $B_{kl}$ are provided in \eqref{eq:ekdcov_a} and \eqref{eq:ekdcov_b} respectively with $a_{kl}(\bC)=\|\bC^T Z_k-\bC^TZ_l\|_2=\|\bm{\beta}^T X_k-\bm{\beta}^TX_l\|_2$ and $b_{kl}=\|\phi(Y_k)-\phi(Y_l)\|_2$. 
If we apply the kernel function $\kappa$ on $\{Y_i\}_{i=1}^n$ to obtain the kernel matrix $\bK\in\mathbb{R}^{n\times n}$ and find its square root $\bM=\bK^{1/2}\in\Rbb^{n\times n}$, we have that $F(\bC)=\nu^2_{\Hcal,n} (\bC^T \bZ, \bY)=\nu^2_{n} (\bC^T \bZ, \bM)$, which follows from $\|\bM(:,k)-\bM(:,l)\|_2^2=\bM(:,k)^T\bM(:,k)+\bM(:,l)^T\bM(:,l)-2\bM(:,k)^T\bM(:,l)=\bK(k,k)+\bK(l,l)-2\bK(k,l)=\|\phi(Y_k)-\phi(Y_l)\|_2^2$. 
Combining with \eqref{eq:alt_emp_dcov}, we are able to reformulate the problem \eqref{eq:empirical_kdcov} as follows:
\begin{align}\label{eq:fd_sdr}
    \max_{\bC\in\Rbb^{p\times d}} F(\bC)=\frac{1}{n^2}\sum^n_{k,l=1}a_{kl}(\bC)\tilde{B}_{kl},\mbox{ s.t. }\bC \in \mbox{St}(d, p),
\end{align}
where $a_{kl}(\bC)=\|\bC^T Z_k-\bC^TZ_l\|_2$ and $\tilde{B}_{kl}$ is given in \eqref{eq:b_dd} with $\tilde{b}_{kl}=\|\bM(:,k)-\bM(:,l)\|_2=\|\phi(Y_k)-\phi(Y_l)\|_2$. 
It is worth noting that the matrix $\tilde{B}$, constructed using $\tilde{b}_{kl} = \|\bM(:, k) - \bM(:, l)\|_2$, is identical to the matrix $B$ constructed using ${b}_{kl} = \|\phi(Y_k) - \phi(Y_l)\|_2$. The distinction in notation is introduced to emphasize their respective contexts: $\tilde{B}$ is employed in the empirical distance covariance, $\nu_n^2(\bC^T \bZ, \bM)$, whereas $B$ is used in the empirical kernel distance covariance, $\nu_{\Hcal, n}^2(\bC^T \bZ, \bY)$. This differentiation helps clarify the role of each matrix in their respective formulations.

\subsubsection*{Kernel Selection}

We provide two kernel functions in Algorithm~\ref{alg:matrix-reg-log}: the Gaussian kernel and the Laplacian kernel, defined respectively as
\[
\kappa_G(Y, Y') = \exp(-\gamma_G \cdot d(Y, Y')^2), \quad \kappa_L(Y, Y') = \exp(-\gamma_L \cdot d(Y, Y')),
\]
where $d(Y, Y')$ represents the distance between $Y$ and $Y'$ in the metric space $\mathcal{Y}$. The choice of distance $d(Y, Y')$ is data-dependent, with the Wasserstein distance used for distributional data and the Frobenius norm for matrix data. These selections align with the definition of $d(Y, Y')$ in metric spaces, ensuring compatibility with various data modalities.

\subsubsection*{Bandwidth Selection}

Following the recommendations in \citep{zhang2023dimension}, the kernel bandwidths $\gamma_G$ and $\gamma_L$ are computed as
\[
\gamma_G = \frac{\rho_Y}{2\sigma_G^2}, \quad \gamma_L = \frac{\rho_Y}{2\sigma_L},
\]
where $\rho_Y$ is a scaling parameter (commonly set to $10$), and $\sigma_G^2$ and $\sigma_L$ are defined as
\[
\sigma_G^2 = \binom{n}{2}^{-1} \sum_{i<j} d(Y_i, Y_j)^2, \quad \sigma_L = \binom{n}{2}^{-1} \sum_{i<j} d(Y_i, Y_j).
\]

These formulas ensure that the bandwidths are tailored to the data's distribution and scale, enhancing the performance of the kernel functions in capturing dependencies. By incorporating the flexibility of different distances and kernels, this framework is adaptable to a wide range of applications and data structures.

The constrained maximization problem in \eqref{eq:fd_sdr} can be solved via various routines such as the conjugate gradient and steepest descent. For completeness, we present the formula for the steepest descent algorithm. Specifically, each (sub)gradient descent step is followed by a projection onto the Stiefel manifold: $\bC^{(\iter+1)}=P_S\Big(\bC^{(\iter)}+\alpha^{(\iter)} \partial_\bC F(\bC^{(\iter)})\Big)$, where $P_S(\cdot)$ denotes the projection on the Stiefel manifold and $\alpha^{(\iter)}$ is the optimal step-size obtained from the backtracking line search~\citep{absil2008optimization}. According to \citep[Proposition 3.4]{absil2012projection}, the projection of $\bC$ onto $\mbox{St}(d,p)$ exists uniquely and can be expressed as $P_S(\bC)=\bU\bV^T$ if the SVD of $\bC\in\mathbb{R}^{p\times d}$ is given by $\bC=\bU\Sigma\bV^T$. Additionally, the explicit formula for the subgradient $\partial_\bC F(\bC)$, where $F(\bC)$ is defined in \eqref{eq:fd_sdr}, is derived as follows:
\begin{equation}\label{eq:gradient}
\partial_\bC F(\bC)=\frac{1}{n^2}
\sum^n_{k,l=1} (\partial[\|\bC^T Z_k-\bC^T Z_l\|_2])\Tilde{B}_{kl} = \frac{1}{n^2}
\sum^n_{k\not=l} \frac{\bC^T(Z_k-Z_l)(Z_k-Z_l)^T}{\|\bC^T(Z_k-Z_l)\|_2}\Tilde{B}_{kl},
\end{equation}
where we assume $Z_k\not=Z_l$ for all $k,l\in[n]$, considering only the corresponding case for the subgradient $\partial[\|\bC^T Z_k-\bC^T Z_l\|_2]$. We summarize the algorithm for Fr\'echet SDR via KdCov (Fd-SDR) in Algorithm \ref{alg:matrix-reg-log}. 

\begin{algorithm}[ht]
  \caption{Fd-SDR}
  \label{alg:matrix-reg-log}
  \begin{algorithmic}[1]
    \STATE {\bf Input:} The samples $\{(X_i,Y_i)\}_{i=1}^n\subset\Rbb^{p}\times\Ycal$, kernel function $\kappa: \calY\times \calY\rightarrow\mathbb{R}$, target dimension: $d(<p)$, maximum number of iterations: $K$, stopping threshold: $\varepsilon$. 
    \STATE {\bf Preparation:} Compute the sample covariance $\hat{\Sigma}_X$ and $\bZ = \hat{\Sigma}^{-\frac12}_X \bX$; compute the kernel matrix $\bK\in\mathbb{R}^{n\times n}$, where $\bK_{ij}=\kappa(Y_i,Y_j),\forall i,j\in[n]$ and $ \bM=\bK^{1/2}\in\mathbb{R}^{n\times n}$.
 
    \STATE {\bf Initialization:} $\bC^{(0)}$.
    \FOR{$\iter=0,1,\ldots,K$}
        \STATE Compute  $\bC^{(\iter+1)}=P_S\Big(\bC^{(\iter)}+\alpha^{(\iter)} \partial_\bC F(\bC^{(\iter)})\Big)$,
        where $\partial_\bC F(\bC)$ is defined by \eqref{eq:gradient}, $\tilde{B}_{kl}$ is given in \eqref{eq:b_dd} with $\tilde{B}_{kl}=\|\bM(:,k)-\bM(:,l)\|$ and $\alpha^{(\iter)}$ is obtained by backtracking line search method.
        \STATE Stop if $\|F(\bC^{(\iter+1)})-F(\bC^{(\iter)})\|_F \leq \varepsilon$.
    \ENDFOR
    \STATE {\bf Output:} Estimated coefficient matrix $\hat{\bm{\beta}} = \hat{\Sigma}^{-\frac12}_X\bC^{(\iter)}$.
  \end{algorithmic}
\end{algorithm}

\subsection{Comparison with other Fr\'echet SDR methods}\label{sec:computation_cost}

In this section, we primarily compare our proposed method, Fd-SDR, with two other Fr\'echet sufficient dimension reduction (SDR) methods, Fr\'echet OPG (FOPG) \citep{zhang2023dimension} and Graphical Weighted Inverse Regression Ensemble (GWIRE) \citep{weng2023sparse}, in terms of the computational complexity and implementation. FOPG demonstrates state-of-the-art performance across general settings for Fr\'echet SDR methods, while GWIRE performs effectively in sparse settings, where an additional assumption is made that the coefficient matrix $\bm{\beta}$ is sparse.

\subsubsection{Computational Comparison}

Our proposed method, detailed in Algorithm~\ref{alg:matrix-reg-log}, demonstrates computational superiority over FOPG and GWIRE. The computational complexity of the preparation step is $O(n^3 + n^2p)$, primarily due to the computation of the kernel matrix and its square root. Step 5 in Algorithm~\ref{alg:matrix-reg-log} involve the optimization over the Stiefel manifold and gradient updates, with a complexity of $O(n^2pd+pd^2)$. In comparison, FOPG \citep[Algorithm 2]{zhang2023dimension} requires solving $n^2$ regression problems in Step 2, which has a computational cost of $O(n^3p)$ and the calculation of $\Lambda^{(t)}$ in Step 3 requires $O(n^2p^2)$, and performing eigenvalue decomposition requires $O(p^2d)$. Consequently, its per-iteration cost is at least $O(n^2p(n+p)+p^2d)$, which is larger than the per-iteration cost of Algorithm~\ref{alg:matrix-reg-log}.  
GWIRE requires a per-iteration cost $O(p^3)$, with an additional cost of $O(np^2)$ for computing the sample covariance. However, its performance highly depends on the a regularization parameter, while our method is parameter-free. The implementation of GWIRE requires choosing a regularization parameter by 5-fold cross-validation over $30$ candidates and is therefore usually much slower than Algorithm~\ref{alg:matrix-reg-log} in practice.

% \subsubsection{Comparison with Fr\'echet sufficient dimension reduction in \citep{zhang2023dimension}}
\subsubsection{Implementation Comparison}\label{sec:compareother}

Our method shares some similarity with the Fr\'echet sufficient dimension reduction (SDR) approach proposed in \citep{zhang2023dimension}. However, since we utilize the distance covariance matrix and base our algorithm on \citep{huang2024robust}, the two methods exhibit key differences. In \citep{zhang2023dimension}, Fr\'echet SDR is applied using $\kappa(Y, \cdot)$ as the response, which distinguishes it from classical SDR techniques, while our method utilizes the feature map $\phi(Y)$ as the response.

% Recall that SDR aims to find $S_{Y|X}$, which represents the central subspace $S$ that satisfying $Y\indep X|P_{S}X$, and for the situation where
% the primary interest is in estimating the regression function, \citep{10.1214/aos/1021379861} introduced
% a weaker form of SDR, the mean dimension reduction subspace, defined as the subspace $S$ such that $\Expect(Y |X) = \Expect(Y |P_SX)$.

Traditional SDR assumes $Y$ is a real-valued response, so these methods can not be directly applied to Fr\'echet SDR, where $Y$ is a metric space-valued response. To handle metric space-valued response, the method in \citep{zhang2023dimension} proposes to find the a class of real-valued functions $\mathcal{F}$ and then use
\[\bigcup \{ S_{E[f(Y)|X]} : f \in \mathcal{F} \}\]
as the estimator of the central subspace. Here $\Expect(Y |X)$ is called mean dimension reduction subspace, a weaker form of SDR introduced in \citep{10.1214/aos/1021379861} and defined as the subspace $S$ such that $\Expect(Y |X) = \Expect(Y |P_SX)$. Note that $S_{E[f(Y)|X]}$ can be estimated using traditional methods as $f(Y)$ is real-valued, so we can turn an SDR method that targets real-valued responses into one that targets metric space-valued responses, and find the Fr\'echet central subspace.

% We know how to estimate the spaces $ S_{E[f(Y)|X]} $ using existing Sufficient Dimension Reduction (SDR) methods because $ f(Y) $ is a real-valued function. By leveraging this ensemble, we can turn an SDR method that targets the central mean subspace into one that targets the Fr\'echet central subspace.

The method in \citep{zhang2023dimension} proposes to use
$\calF = \{\kappa(\cdot, y) : y \in \Omega_Y\}$, where $\kappa$ is a RHKS, and the method implemented as follows. Let $F_{XY}$ is the empirical distribution of observations and $F_Y$ is the empirical distribution of $Y$, then the proposed estimator of the Fr\'echet central subspace $\mathcal{S}_{\bY\vert \bX}$ is 
\begin{equation}
    M(F_{XY})= \int _{\Omega_Y} M_0(F_{XY},\kappa(\cdot,y))\ dF_Y(y) 
    \label{eq:M},
\end{equation}
where $M_0(F_{XY},\kappa(\cdot,y))$ is the estimator of classical SDR central subspace with responses $$\kappa(Y_1,y), \cdots, \kappa(Y_n,y).$$

%\subsubsection{CMS-Ensemble}

\section{Theoretical Guarantees}\label{sec:theory}

In this section, we present three theoretical results for the proposed Fd-SDR method outlined in Algorithm~\ref{alg:matrix-reg-log}: 1) The justification for performing Fr\'echet SDR using feature map-based responses $\{\phi(Y_i)\}_{i=1}^n$ is established in Theorem~\ref{thm:equivalent}; 2) The statistical consistency of Fd-SDR under the model specified in \eqref{eq:model} is provided in Theorem~\ref{thm:consistency}; 3) The algorithmic convergence result for Fd-SDR is presented in Theorem~\ref{thm:convergence}.

Our theoretical guarantee is based on the concept of CMS-Ensemble, introduced by Zhang et al. in \citep{zhang2023dimension} as follows. 
\begin{definition}[Definition of CMS-Ensemble]
The Central Mean Space ensemble (CMS-ensemble) is defined as a family $ \mathcal{F} $ that is rich enough so that
\begin{equation}\label{eq:CMS}
S_{Y|X} = \bigcup \{ S_{E[f(Y)|X]} : f \in \mathcal{F} \}.    
\end{equation}
Recall that $ S_{E[f(Y)|X]} $  is the mean dimension reduction subspace defined in Section \ref{sec:compareother}.
\end{definition}

This definition implies that the ensemble $ \mathcal{F} $ is a collection of functions rich enough  such that the union of the spaces $ S_{E[f(Y)|X]} $ captures the full central subspace $ S_{Y|X} $. In addition, \citep{zhang2023dimension} proves that for a large range of kernels and $\Omega_Y$, the family $\calF = \{\kappa(\cdot, y) : y \in \Omega_Y\}$ constitutes a CMS-ensemble. In particular, they prove that this holds as long as  (1) $\kappa$ is a bounded, compact-convergence (cc)-universal kernel (cc-universal implies that the RKHS associated with 
$\kappa$ is sufficiently rich to approximate any compactly supported continuous function in $\ell_\infty$ norm, and we refer readers to    \citep{JMLR:v7:micchelli06a,JMLR:v12:sriperumbudur11a} for formal definition), and (2) $P_Y$ is a regular probability measure. We refer the interested reader to \citep{zhang2023dimension}  for more technical details.

Theorem~\ref{thm:equivalent} shows that when the kernel is a CMS-ensemble, the central subspace defined by the original responses $\bY$ coincides with that defined by their feature map representations $\phi(\bY)$. It guarantees the performance of the proposed Fd-SDR Algorithm, as its objective function is derived from dCov-based SDR by replacing response $Y$ with its feature map $\phi(Y)$, see \eqref{eq:empirical_kdcov_F}.  Following \citep{zhang2023dimension}, a wide range of kernels fall into the CMS-ensemble category, including Gaussian and Laplacian kernels in various spaces such as the Euclidean space, Wasserstein space, the space of symmetric positive definite matrices, and the sphere. 
Notably, this result pertains to the central subspace and is independent of the specific algorithm used, ensuring broad applicability.

\begin{theorem}[Theoretical guarantee for Fd-SDR  Algorithm and feature map-based SDR]\label{thm:equivalent}
Let $\calS_{Y|X}$ represents the central subspace of random variables $(X,Y)$, and $\kappa$ is a RKHS. If the family of functions $\{\kappa(\cdot, \bY):\bY\in\calY\}$ is a CMS-ensemble, then $\calS_{\phi(Y)|X}=\calS_{Y|X}$, where $\phi$ is the feature map induced by the RKHS $\kappa$. 
\end{theorem}
\begin{proof}[Proof of Theorem~\ref{thm:equivalent}]
Let $S=\calS_{Y|X}$, then by definition, $Y\independentt X| P_SX$, and as a result, $\phi(Y)\independentt X|P_SX$ and 
\begin{equation}\calS_{\phi(Y)|X}\subseteq \calS_{Y|X}.\label{eq:equivalent1}\end{equation}

As $\kappa(Y,Y_0)=\langle\phi(Y),\phi(Y_0)\rangle$, we have that for all $Y_0\in\calY$,
\[
\calS_{\kappa(Y,Y_0)|X}\subseteq \calS_{\phi(Y)|X}
\]
and
\begin{equation}\label{eq:equivalent2}
S_{Y|X}=\mathrm{span}(\calS_{\kappa(Y,Y_0)|X}:Y_0\in\calY)\subseteq \calS_{\phi(Y)|X},
\end{equation}
where the first equality follows from the assumption that the family of functions $\{\kappa(\cdot, \bY):\bY\in\calY\}$ is a CMS-ensemble. Combining \eqref{eq:equivalent1} and \eqref{eq:equivalent2}, the theorem is proved.
\end{proof}

Our second result establishes the consistency of Algorithm~\ref{alg:matrix-reg-log}. Following \citep{huang2024robust}, we consider a  model with a general noise term
\begin{equation}\bY=g(\bm{\bbeta}_0^T\bX,\epsilon)=g(\bC_0^T\bZ,\epsilon),\label{eq:model}\end{equation} 
where $\bm{\bbeta}_0$ is a $p\times d$ orthogonal matrix, $g(\cdot)$ is an unknown link function, $\bC_0 = \hat{\Sigma}^{\frac12}_X \bm{\bbeta}_0$, and $\bZ = \hat{\Sigma}^{-\frac12}_X \bX$, and $\epsilon$ is independent of $\bZ$. This model includes the model from \citep{xia2002adaptive} that
 $\bY=g(\bm{\bbeta}_0^T \bX)+\epsilon$ is a special example.

The theorem on  the asymptotic properties of our estimator $\bC$ up to some rotation matrix $\bQ$ in \citep[Proposition 3.1]{huang2024robust} also applies to our setting, which implies the asymptotic property of the estimated central subspace.

\begin{theorem}[Consistency guarantee for Fd-SDR  Algorithm]\label{thm:consistency}
Under model \eqref{eq:model}, and let $\bC \in \mathbb{R}^{d\times p}$ be a basis of the central subspace $S_{Y\mid X}$ with $\bC^T\Sigma_{\bX}\bC=\Ibf_d$. 
Suppose $P^T_{\bC (\Sigma_{\bX})} \bX \indep Q^T_{\bC(\Sigma_{\bX} )}\bX$ 
and the support of $\bX \in \mathbb{R}^{d\times p}$, say $S$, is a compact set. In addition, assume that there exists $\bC'\in\mathbb{R}^{(p-d)\times p}$ such that $[\bC,\bC']^T\Sigma_{\bX}[\bC,\bC']=I_p$ and $\bC^T\bX$ is independent of $\bC'^T\bX$.
Let $\hat{\bC}_n=\arg\min_{\bC^T\Sigma_{\bX}\bC=I_d}\nu^2_n(\bC^T\bX,\bY)$, then there exists a rotation matrix $\bQ$: $\bQ^T\bQ=I_d$ such that
$\hat{\bC}_n\stackrel{P}{\to}\bC\bQ$ (convergence in probability) as $n \to\infty$, with a rate of $1/\sqrt{n}$:  $\min_{\bQ}\|\hat{\bC}_n-\bC\bQ\|_F=O_p(1/\sqrt{n}).$

\label{prop1}
\end{theorem}
\begin{proof}[Proof of Theorem~\ref{prop1}]
First, we note that our method is  rSDR~\citep{huang2024robust} with responses being $\phi(Y)$. Since our response $\phi(Y)=\phi(g(\bm{\bbeta}_0^T\bX,\epsilon))$ is also a function of $\bm{\bbeta}_0^T\bX$ and $\epsilon$, the proof of convergence follows from ~\citep[Proposition 3.1]{huang2024robust}.

%It remains to prove the consistency rate of $1/\sqrt{n}$...

Next, we prove the convergence rate $ 1/\sqrt{n} $. The minimization of the empirical risk function implies that the difference $ \hat{\bC}_n - \bC \bQ $ becomes small as $ n \to \infty $. More precisely, using Taylor expansion around the true parameter $ \bC $, we have
\[
\nu_n^2(\hat{\bC}_n^T \bX, \bY) - \nu_n^2(\bC^T \bX, \bY) = O_p(1/{n}).
\]
This indicates that the difference $ \hat{\bC}_n - \bC \bQ $ is of order $ O_p(1/\sqrt{n}) $.
Finally, to conclude that $ \|\hat{\bC}_n - \bC \bQ\|_F = O_p(1/\sqrt{n}) $, we apply standard results from matrix perturbation theory \citep{stewart1990matrix}. Since $ \hat{\bC}_n $ is constrained by $ \hat{\bC}_n^T \Sigma_{\bX} \hat{\bC}_n = \Ibf_d $, the difference between $ \hat{\bC}_n $ and $ \bC \bQ $ is bounded in the Frobenius norm by the rate of convergence of the empirical risk function, which yields the desired result.
Thus, we have established that
\[
\min_{\bQ} \|\hat{\bC}_n - \bC \bQ\|_F = O_p(1/\sqrt{n}).
\]

\end{proof}

Similar to \citep[Proposition 3.1]{huang2024robust}, Theorem~\ref{prop1} requires an additional assumption related to the decomposition of $ X $ into two independent components. Further discussion on this condition can be found in \citep[Section 3.2]{sheng2013direction}. For instance, this assumption holds when $ X $ follows a normal distribution \citep{zhang2015direction} and is asymptotically satisfied when $ p $ is large, as shown in \citep{hall1993almost}.

Our final result establishes the algorithmic convergence of Algorithm~\ref{alg:matrix-reg-log}. 
Since Fd-SDR is based on the algorithm proposed in \citep{huang2024robust}, its theoretical guarantee in \citep[Theorem 3.2]{huang2024robust} can be easily generalized to our setting, which states that any accumulation point is a stationary point of the objective function.
\begin{theorem}[Convergence of Fd-SDR  Algorithm]\label{thm:convergence}
(a) Any accumulation point of the sequence $\left\{\hat{\bC}^{(t)}\right\}_{t\geq 0}$ generated by the proposed algorithm converges is a stationary point of $F_{\eta}(\bC)$ over the Stiefel manifold, the set of all orthogonal matrices of size $\mathbb{R}^{p\times d}$.% and it converges to a stationary point of the objective function
%$\nu^2_n(\hat{\bC}^{(t)\,T}\bX,\bY,\alpha)$.

%
(b) If in addition, the global maximizer $\hat{\bC}$ it is the unique stationary point in its neighborhood $\mathcal{N}$, and  $F_{\eta}(\bC)-F_{\eta}(\hat{\bC})\leq -c\|\bC-\hat{\bC}\|_F^2$ for any $\bC$ in $\mathcal{N}$ and some $c>0$. Then when the initialization $\hat{\bC}^{(0)}$ is sufficiently close to $\hat{\bC}$, the sequence $\left\{\hat{\bC}^{(t)}\right\}_{t\geq 0}$ converges to $\hat{\bC}$. 
\end{theorem}

\section{Simulation Studies}\label{sec:syn}
In this section, we compare the performance of the proposed Fd-SDR method with FOPG~\citep{zhang2023dimension} and GWIRE~\citep{weng2023sparse} in terms of accuracy across the different synthetic settings described in Sections~\ref{sec:scenario_I}-\ref{sec:scenario_III}. Additionally, we provide a comparison of the computational efficiency in Section~\ref{sec:time_comparison}.

\subsection{Scenario I: SDR for distributions}\label{sec:scenario_I}

We set $\bbeta_1 = (1,1,0,\ldots,0)^T$, $\bbeta_2 = (0,\ldots,0,1,1)^T$, $\bbeta_3 = (1,2,0,\ldots,0,2)^T$, and $\bbeta_4 = (0,0,1,2,2,0,\ldots,0)^T$ in $\mathbb{R}^p$. Additionally, we let $Y \sim \mathcal{N}(\mu_Y, \alpha \sigma_Y^2)$, where $\mu_Y$ and $\sigma_Y$ are random variables depending on $X$. Denote the Fr\'echet central subspace by $\mathcal{S}$ and its estimation with $\hat{\mathcal{S}}$. The following four models are considered:
\begin{enumerate}[label=(\arabic*)]
    \item $\mu_Y \sim \Ncal(\exp(\bbeta_1^T X),0.1)$, $\sigma_Y = 0.5$, $\alpha=1$ and $\Scal=[\bbeta_1]$.
    \item $\mu_Y \sim \Ncal(\exp(\bbeta_1^T X),0.1)$, $\sigma_Y = \exp(\bbeta_2^T X)$ with truncated range $(10^{-1},10)$, $\alpha=1$ and $\Scal=[\bbeta_1,\bbeta_2]$.
    \item $\mu_Y \sim \Ncal(3(\bbeta_3^T X), 0.5^2)$, $\sigma_Y = \text{Gamma} ((2+2\bbeta_4^T X)^2/\nu , \nu /(2+2 \bbeta_4^T X) )$ with truncated range $(10^{-1},10)$ and $\nu = 0.5$, $\alpha=0.2/0.4$ and $\Scal=[\bbeta_3,\bbeta_4]$.
    \item $\mu_Y \sim \Ncal(3\sin(\bbeta_3^T X), 0.5^2)$, $\sigma_Y = \text{Gamma} ((2+2\bbeta_4^T X)^2/\nu , \nu /(2+2 \bbeta_4^T X) )$ with truncated range $(10^{-1},10)$ and $\nu = 0.5$, $\alpha=0.2/0.4$ and $\Scal=[\bbeta_3,\bbeta_4]$.
\end{enumerate}
We consider three different approaches for generating the predictor $X$:
\begin{enumerate}[label=(\alph*)]
    \item $X \sim \Ncal(0,\Ibf_p)$.
    \item Generate $U_1,U_2,\ldots,U_p$ from AR(1) model with mean 0 and covariance $\Sigma = (0.5^{|i-j|})_{i,j}$, then set $X=(\sin(U_1),|U_2|,U_3,\ldots,U_p)$.
    \item Generate $U_1,U_2,\ldots,U_p$ from AR(1) model with mean 0 and covariance $\Sigma = (0.5^{|i-j|})_{i,j}$, then set $X=[\Phi(U_1),\ldots,\Phi(U_p)]$, where $\Phi(\cdot)$ is the c.d.f. of the standard normal distribution.
\end{enumerate}
We generate the set of samples $\{(X_i,Y_i)\}_{i=1}^n$ following models above. For  Models (1) and (2), we use approaches (a) and (b) to generate $X$. For Models (3) and (4), approach (c) is used to generate $X$. The estimated Fr\'echet central subspace $\hat{\Scal}$ is obtained using Fd-SDR, FOPG, and GWIRE with the Gaussian kernel $\kappa_G(Y, Y') = \exp(-\gamma_G \cdot d(Y, Y')^2)$. Details on the choice of bandwidth and related specifications can be found in Section~\ref{sec:gen_fsdr}. The estimation error is defined by the Frobenius norm of the difference between the projectors on $\mathcal{S}$ and $\hat{\mathcal{S}}$:  $e(\mathcal{S},\hat{\mathcal{S}})=||P_{\mathcal{S}}-P_{\hat{\mathcal{S}}}||_{F}=|| B(B^T B)^{-1}B^T-\hat{B}(\hat{B}^T \hat{B})^{-1}\hat{B}^T||_{F}$. We repeat the simulation 100 times for two cases, $(n,p)=(200,10)$ and $(n,p)=(400,20)$, and report the mean errors as well as the standard deviations in Table~\ref{tab:sceI}. The box plots comparing the errors of the three estimators are shown in Figure~\ref{fig:combined_errors_scenario_1}. 

\begin{table}[ht]
\centering
\renewcommand{\arraystretch}{1.2} % Increase space between rows
\begin{tabular}{c|l|ccc}
\toprule
$(n,p)$                     & Model               & Fd-SDR              & FOPG                & GWIRE               \\ \midrule
\multirow{8}{*}{(200,10)} & (1-a)               & 0.09(0.03)          & 0.16(0.08)          & \textbf{0.07(0.05)} \\
                          & (1-b)               & \textbf{0.09(0.03)} & 0.11(0.04)          & 0.20(0.06)          \\
                          & (2-a)               & \textbf{0.18(0.03)} & 0.25(0.12)          & 0.23(0.07)          \\
                          & (2-b)               & \textbf{0.25(0.07)} & 0.27(0.06)          & 0.71(0.20)          \\
                          & (3-c), $\alpha=0.2$ & \textbf{0.45(0.17)} & 1.39(0.03)          & 0.89(0.17)          \\
                          & (3-c), $\alpha=0.4$ & \textbf{0.30(0.06)} & 0.70(0.36)          & 0.67(0.12)          \\
                          & (4-c), $\alpha=0.2$ & \textbf{0.36(0.07)} & 1.29(0.15)          & 0.72(0.14)          \\
                          & (4-c), $\alpha=0.4$ & 0.34(0.07)          & \textbf{0.29(0.07)} & 0.62(0.08)          \\
                          \midrule
\multirow{8}{*}{(400,20)} & (1-a)               & 0.09(0.02)          & 0.22(0.05)          & \textbf{0.05(0.04)} \\
                          & (1-b)               & \textbf{0.09(0.02)} & 0.18(0.03)          & 0.20(0.05)          \\
                          & (2-a)               & 0.22(0.03)          & 0.29(0.09)          & \textbf{0.13(0.06)} \\
                          & (2-b)               & \textbf{0.26(0.05)} & 0.33(0.05)          & 0.64(0.15)          \\
                          & (3-c), $\alpha=0.2$ & \textbf{0.40(0.08)} & 1.40(0.03)          & 0.86(0.17)          \\
                          & (3-c), $\alpha=0.4$ & \textbf{0.30(0.04)} & 0.92(0.35)          & 0.64(0.09)          \\
                          & (4-c), $\alpha=0.2$ & \textbf{0.35(0.05)} & 1.36(0.10)          & 0.69(0.10)          \\
                          & (4-c), $\alpha=0.4$ & 0.34(0.05)          & \textbf{0.29(0.05)} & 0.59(0.05) \\
                          \bottomrule
\end{tabular}
\caption{Errors of Fd-SDR, FOPG and GWIRE in mean and standard deviation (in parentheses) for Scenario I.}
\label{tab:sceI}
\end{table}

\begin{figure}[ht]
    \centering
    % First row: Case 3
    \begin{subfigure}[b]{0.4\textwidth}
        \centering
        \includegraphics[width=\textwidth]{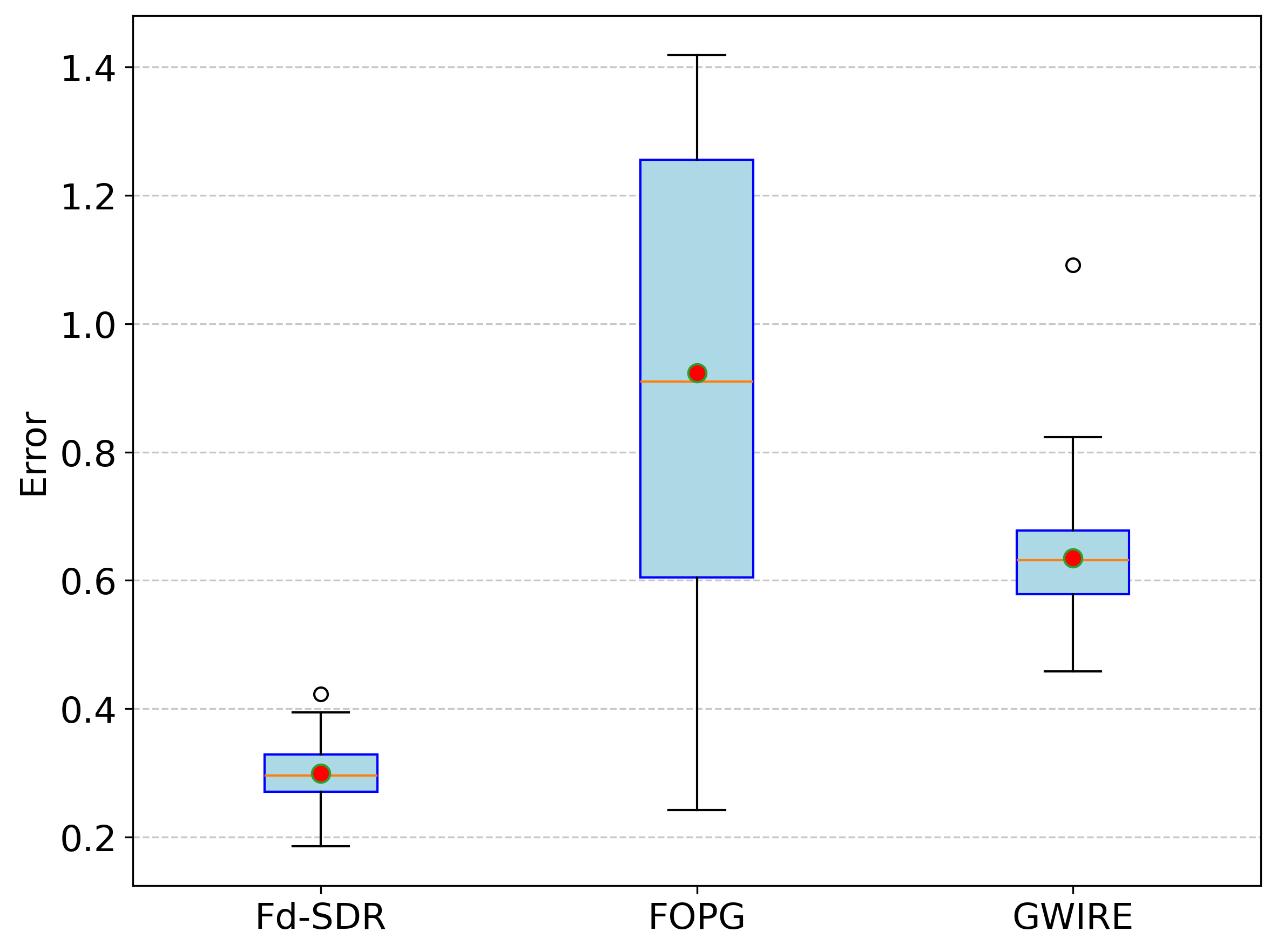}
        \caption{Model (3-c), $\alpha=0.4$}
        \label{fig:case3_alpha0.4}
    \end{subfigure}
    % \hfill
    \begin{subfigure}[b]{0.4\textwidth}
        \centering
        \includegraphics[width=\textwidth]{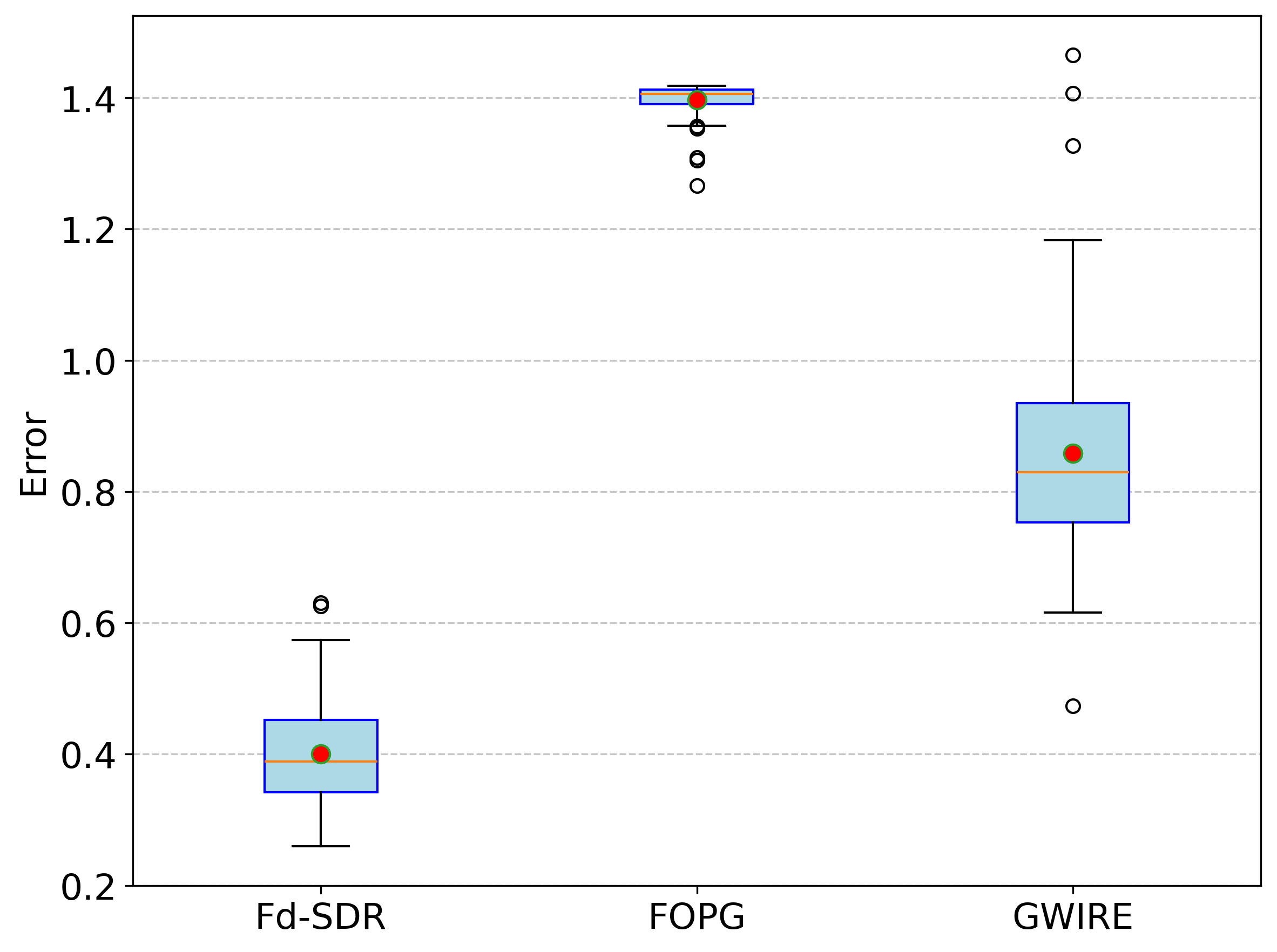}
        \caption{Model (3-c), $\alpha=0.2$}
        \label{fig:case3_alpha0.2}
    \end{subfigure}
    
    \vspace{1em} % Add some vertical spacing
    
    % Second row: Case 4
    \begin{subfigure}[b]{0.4\textwidth}
        \centering
        \includegraphics[width=\textwidth]{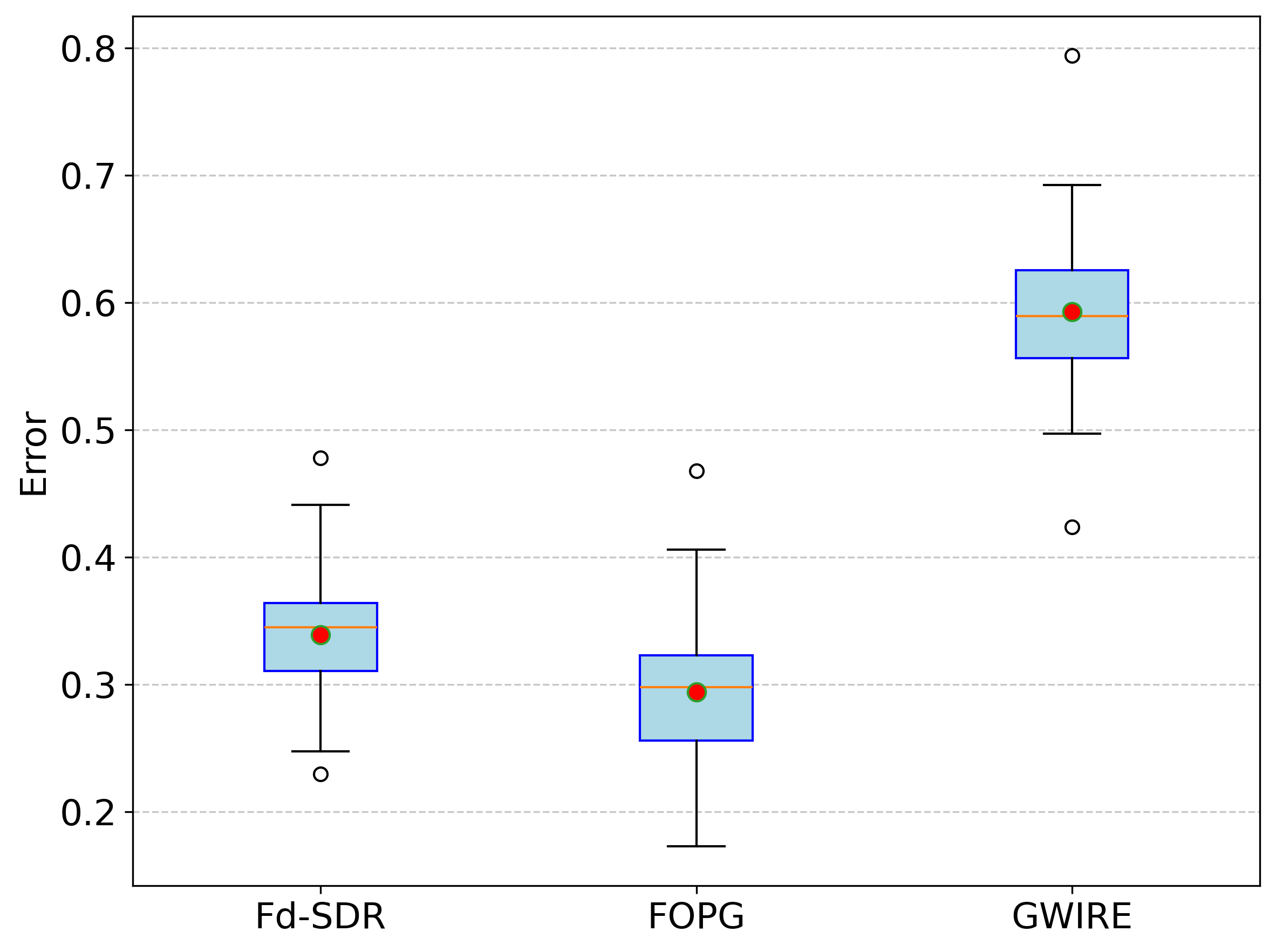}
        \caption{Model (4-c), $\alpha=0.4$}
        \label{fig:case4_alpha0.4}
    \end{subfigure}
    % \hfill
    \begin{subfigure}[b]{0.4\textwidth}
        \centering
        \includegraphics[width=\textwidth]{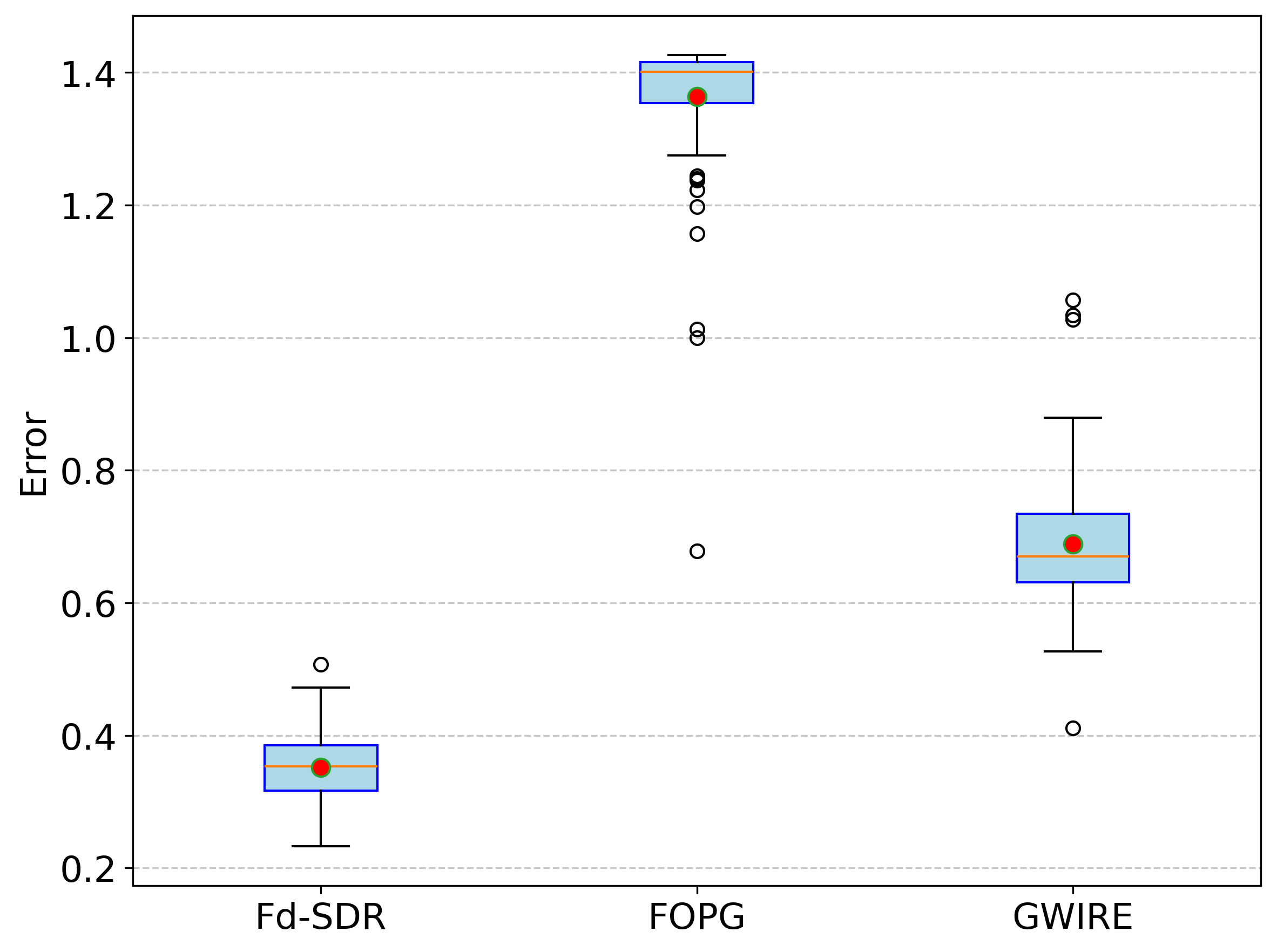}
        \caption{Model (4-c), $\alpha=0.2$}
        \label{fig:case4_alpha0.2}
    \end{subfigure}    
    \caption{Box plots of Fd-SDR, FOPG and GWIRE in Models (3-c) and (4-c) with $n=400$ from Scenario I.}
    \label{fig:combined_errors_scenario_1}
\end{figure}

The results in Table~\ref{tab:sceI} demonstrate that Fd-SDR achieves smaller estimation errors than FOPG and GWIRE in most models, with a few exceptions where Fd-SDR remains competitive. Especially in Model (3-c) and (4-c), Fd-SDR can estimate the accurate subspace while other two estimators fail. In these two models, the errors for both FOPG and GWIRE increase significantly as the value of $\alpha$ decreases while Fd-SDR still leads to a good estimation. We present the box plots in Figure~\ref{fig:combined_errors_scenario_1} for these two models with $n=400$. The plot indicates that Fd-SDR has smaller mean and variance over the 100 repeated simulations and has few outliers than FOPG and GWIRE. Moreover, we notice that GWIRE shows the smallest estimation errors in Model (1-a) for both of $(n,p)=(200,10)$ and $(n,p)=(400,20)$. However, GWIRE fails and has large errors especially in Model (3-c) and (4-c), which is mainly because GWIRE assumes the sparsity of the subspace, an assumption that these two models do not satisfy.
The results highlight that Fd-SDR is able to estimate the Fr\'echet central subspace and has better performance compared to the existing FSDR methods across various models, making it a reliable method in general.

\subsection{Scenario II: SDR for matrices}
In this synthetic dataset, we use similar settings to those in Scenario I, but with a modification in how the response $Y$ is generated. Specifically, the vectors $ \bbeta_1 $ and $ \bbeta_2 $, as well as the approaches (a) and (b) for generating $ X $ from Scenario I, are retained. We set $ \log(Y) \sim \Ncal_{q \times q}(\log(D(X)), 0.5^2\Ibf_q) $, where two models of $ D(X) $ are defined:
\begin{enumerate}[label=(\arabic*)]
    \item $D(X) = \begin{bmatrix}
               1 & \rho(X) \\
               \rho(X) & 1
               \end{bmatrix}$, where $\rho(X) = 0.8\cos{(\bbeta_1^T X)}$, and $\Scal=[\bbeta_1]$.
    \item $ D(x) = \begin{bmatrix}
                1 & \rho_1(X) & \rho_2(X)\\
                \rho_1(X) & 1 &\rho_1(X)  \\
                \rho_2(X) & \rho_1(X) & 1
        \end{bmatrix}$ , where $\rho_1(X) = 0.8\cos{(\bbeta_1^T X)}, \rho_2(X)= 0.8\sin(\bbeta_2^T X)$, and $\Scal=[\bbeta_1,\bbeta_2]$.
\end{enumerate}
Here $\Ncal_{q\times q}(M,\Sigma)$ with symmetric mean $M$ and covariance $\Sigma$ refers to the symmetric matrix variate normal distribution~\citep{schwartzman2006random}. The p.d.f. of $Z\sim \Ncal_{q\times q}(M,\Sigma)$ is given by
\begin{align*}
    f(z;M,\Sigma)=\frac{1}{(2\pi)^{p/2}\vert \Sigma\vert^{(q+1)/2}}\exp(-\frac{1}{2}\text{tr}((z-M)\Sigma^{-1})^2),
\end{align*}
where $p=q(q+1)/2$.

Similar to Scenario I, we estimate Fr\'echet central subspace using Fd-SDR, FOPG, and GWIRE with the Gaussian kernel $\kappa_G(Y, Y') = \exp(-\gamma_G \cdot d(Y, Y')^2)$. The simulation is repeated 100 times, and the means and standard deviations of the estimation errors for the three estimators are summarized in Table~\ref{tab:sceII}. The results of this table show that Fd-SDR outperforms the other two methods and the errors in Model (2) are noticeable smaller than FOPG and GWIRE with the only exception in Model (1-b). Even in this model, Fd-SDR only has slightly large errors. We remark that GWIRE is not applicable (NA) for Model (1-a) as it fails to select the regularization parameter.

\begin{table}[ht]
\centering
\renewcommand{\arraystretch}{1.2} % Increase space between rows
\begin{tabular}{c|c|ccc}
\toprule
$(n,p)$                     & Model               & Fd-SDR              & FOPG                & GWIRE               \\ \midrule
\multirow{4}{*}{(200,10)} & (1-a)                      & \textbf{0.24(0.29)}      & 0.30(0.33)             & NA                        \\
                          & (1-b)                      & 0.22(0.08)               & 0.24(0.10)             & \textbf{0.20(0.06)}     \\
                          & (2-a)                      & \textbf{0.30(0.07)}      & 0.37(0.23)             & 1.35(0.14)              \\
                          & (2-b)                      & \textbf{0.43(0.12)}      & 0.67(0.37)             & 0.72(0.19)              \\
                          \midrule
\multirow{4}{*}{(400,20)} & (1-a)                      & \textbf{0.36(0.44)}      & 0.52(0.49)             & NA                        \\
                          & (1-b)                      & 0.24(0.07)               & 0.24(0.10)             & \textbf{0.20(0.07)}     \\
                          & (2-a)                      & \textbf{0.45(0.35)}      & 0.66(0.47)             & 1.41(0.04)              \\
                          & (2-b)                      & \textbf{0.44(0.07)}      & 0.74(0.34)             & 0.69(0.14)  \\
                          \bottomrule
\end{tabular}
\caption{Errors of Fd-SDR, FOPG and GWIRE in mean and standard deviation (in parentheses) for Scenario II.}
\label{tab:sceII}
\end{table}

% \begin{table}[!ht]
%     \centering
%     \renewcommand{\arraystretch}{1.5} % Increase space between rows
%     \begin{tabular}{c|S[table-format=1.3]lS[table-format=1.3]lS[table-format=1.3]l}
%     \hline
%     \textbf{Method} & \multicolumn{2}{c}{\textbf{Fd-SDR}} & \multicolumn{2}{c}{\textbf{FOPG}} & \multicolumn{2}{c}{\textbf{GWIRE}} \\
%     \hline
%     \multicolumn{7}{c}{\textbf{n = 200, p = 10}} \\
%     \hline
%     1-a & \textbf{0.244} & (0.288) & 0.300 & (0.331) & \multicolumn{2}{c}{\textit{NA}} \\
%     1-b & 0.223 & (0.076) & 0.240 & (0.104) & \textbf{0.203} & (0.058) \\
%     2-a & \textbf{0.298} & (0.066) & 0.374 & (0.232) & 1.348 & (0.143) \\
%     2-b & \textbf{0.432} & (0.124) & 0.672 & (0.370) & 0.716 & (0.186) \\
%     \hline
%     \multicolumn{7}{c}{\textbf{n = 400, p = 20}} \\
%     \hline
%     1-a & \textbf{0.359} & (0.437) & 0.519 & (0.491) & \multicolumn{2}{c}{\textit{NA}} \\
%     1-b & 0.239 & (0.065) & 0.235 & (0.102) & \textbf{0.202} & (0.069) \\
%     2-a & \textbf{0.448} & (0.354) & 0.660 & (0.468) & 1.409 & (0.040) \\
%     2-b & \textbf{0.442} & (0.068) & 0.739 & (0.338) & 0.686 & (0.141) \\
%     \hline
%     \end{tabular}
%     \caption{Error comparisons (standard
% error in parentheses) in Scenario II (with $n = 200$ and $n = 400$).}
%     \label{tab:sceII}
% \end{table}

\subsection{Scenario III: SDR for spherical data}\label{sec:scenario_III}
In this simulation, we use similar settings to those in Scenario I, but with modifications to the subspace vectors and the way the response $Y$ is generated. The approaches (a) and (c) for generating $X$ from Scenario I are applied. Let   $\bbeta_1 = \frac{1}{\sqrt{3}}(x_1,x_2,x_3,0,\ldots,0)^T$,\ $\bbeta_2 = \frac{1}{\sqrt{3}}(0,\ldots,0,x_{p-2},x_{p-1},x_p)^T$ and $\bbeta_3 = (0,1,0,\ldots,0)^T$  with $x_i\overset{i.i.d.}{\sim} \Ncal(0,1)$. We generate the spherical data $Y$ from the following three models:
\begin{enumerate}[label=(\arabic*)]
    \item $Y$ is generated by $Y = \cos{(||\varepsilon(X)||)}m(X)+\sin{(||\varepsilon(X)||)}\frac{\varepsilon(X)}{||\varepsilon(X)||}$ where $m(X),\varepsilon(X)$ are given by
    \begin{align*}
        m(X) &= (\cos{(\pi \bbeta_1^T X)},\sin{(\pi \bbeta_1^T X)}) \\
        \varepsilon(X) &= (-\delta\sin{(\pi \bbeta_1^T X)},\delta\cos{(\pi \bbeta_1^T X)})
    \end{align*}
    with $\delta\sim N(0,0.2^2)$. The underlying subspace is $\mathcal{S} = [\bbeta_1]$.
    
    \item $Y$ is generated by $Y = \cos{(||\varepsilon||)}m(X)+\sin{(||\varepsilon||)}\frac{\varepsilon}{||\varepsilon||}$ where $m(X),\varepsilon$ are given by
    \begin{align*}
        m(X) &= ((1-(\bbeta_3^T X)^2)^{1/2}\cos{(\pi \bbeta_1^T X )},(1-(\bbeta_3^T X)^2)^{1/2}\sin{(\pi \bbeta_1^T X)},\bbeta_3^T X) \\
        \varepsilon &= \delta_1v_1+\delta_2v_2,
    \end{align*}
    where $\delta_{1,2}\sim N(0,0.2^2)$ and $v_1,v_2$ are orthogonal basis of the tangent space $T_{m(X)}\Sbb^2$. The underlying subspace is $\mathcal{S} = [\bbeta_1,\bbeta_3]$.

    \item $Y$ is generated by 
    \begin{align*}
        Y = (\sin{(\bbeta_1^T X+ \delta_{1})}\sin{(\bbeta_2^T X+\delta_{2})},\sin{(\bbeta_1^T X+\delta_{1})}\cos{(\bbeta_2^T X+\delta_{2})},\cos{(\bbeta_1^T X+\delta_{1})})
    \end{align*}
    with $\delta_{1,2}\sim\Ncal(0,0.2^2)$. The underlying subspace is $\mathcal{S} = [\bbeta_1,\bbeta_2]$.
\end{enumerate}
By applying the Gaussian kernel function, we estimate Fr\'echet central subspace using Fd-SDR, FOPG, and GWIRE. The simulation is repeated 100 times, and the means and standard deviations of the estimation errors for the three estimators are summarized in Table~\ref{tab:sceIII}.

The simulation results in Table \ref{tab:sceIII} demonstrate that Fd-SDR is competitive with FOPG and GWIRE in Model (1-a), (1-b) and (2-a), while achieving smaller errors in other three models for both cases of $n=200, p=10$ and $n=400, p=20$. We conclude that for the spherical data that generated in this scenario Fd-SDR can estimate the Fr\'echet subspace well and suit for the general settings of Fr\'echet SDR. 

\begin{table}[ht]
\centering
\renewcommand{\arraystretch}{1.2} % Increase space between rows
\begin{tabular}{c|c|ccc}
\toprule
$(n,p)$                     & Model               & Fd-SDR              & FOPG                & GWIRE               \\ \midrule
\multirow{6}{*}{(200,10)} & (1-a)                      & 0.12 (0.06)       & \textbf{0.10 (0.05)}     & 0.23 (0.10)               \\
                          & (1-b)                      & 0.15 (0.10)                & \textbf{0.13 (0.09)}     & 0.49 (0.14)      \\
                          & (2-a)                      & 0.32 (0.18)       & \textbf{0.30 (0.18)}     & 0.33 (0.23)               \\
                          & (2-b)                      & \textbf{0.39 (0.28)}       & 0.40 (0.29)              & 0.69 (0.26)               \\
                          & (3-a)                      & \textbf{0.23 (0.14)}       & 0.29 (0.30)              & 0.79 (0.44)               \\
                          & (3-b)                      & \textbf{0.26 (0.10)}       & 0.31 (0.26)              & 0.85 (0.24)      \\
                          \midrule
\multirow{6}{*}{(400,20)} & (1-a)                      & 0.12 (0.06)       & \textbf{0.11 (0.05)}     & 0.21 (0.10)               \\
                          & (1-b)                      & 0.16 (0.17)       & \textbf{0.15 (0.17)}     & 0.49 (0.15)               \\
                          & (2-a)                      & 0.37 (0.28)                & 0.39 (0.31)              & \textbf{0.30 (0.31)}      \\
                          & (2-b)                      & \textbf{0.41 (0.23)}       & 0.46 (0.29)              & 0.66 (0.23)               \\
                          & (3-a)                      & \textbf{0.25 (0.15)}       & 0.39 (0.37)              & 0.65 (0.45)               \\
                          & (3-b)                      & \textbf{0.28 (0.09)}       & 0.45 (0.36)              & 0.90 (0.28)  \\
                          \bottomrule
\end{tabular}
\caption{Errors of Fd-SDR, FOPG and GWIRE in mean and standard deviation (in parentheses) for Scenario III. Here Fd-SDR uses the $\hat{\Scal}^{\text{GWIRE}}$ as the initialization.}
\label{tab:sceIII}
\end{table}

\subsection{Running Time Comparison}\label{sec:time_comparison}
As outlined in Section~\ref{sec:computation_cost}, we previously discussed the computational complexity of the proposed method, Fd-SDR, and demonstrated its computational advantage over existing Fr\'echet SDR methods such as FOPG and GWIRE. In this section, we provide a numerical verification of this advantage by comparing the running times of Fd-SDR, FOPG, and GWIRE using the settings in Sections~\ref{sec:scenario_I}-\ref{sec:scenario_III}. All simulations were performed in Python on a machine equipped with an AMD Ryzen 5 5600X 6-core processor.

Table \ref{tab:time_comparison} presents the running times for each scenario with $(n,p) = (400,20)$. The results demonstrate that Fd-SDR is considerably faster than both FOPG and GWIRE across all settings, achieving speeds 10 to 20 times faster than FOPG and 50 to 100 times faster than GWIRE. These findings confirm that Fd-SDR is significantly more efficient than FOPG and GWIRE in all cases. It is worth noting that Fd-SDR computes the kernel matrix only once outside the iterative process, while the other two methods require the kernel matrix to be recalculated in each iteration.

\begin{table}[ht]
\centering
\renewcommand{\arraystretch}{1.2} % Increase space between rows
\begin{tabular}{c|l|ccc}
\toprule
                    & Model               & Fd-SDR              & FOPG                & GWIRE               \\ \midrule
\multirow{8}{*}{Scenario I}   & (1-a)                      & 0.07 (0.01)                & 2.52 (0.17)              & 3.37 (2.50)               \\
                              & (1-b)                      & 0.06 (0.01)                & 2.62 (0.12)              & 12.24 (0.54)              \\
                              & (2-a)                      & 0.03 (0.01)                & 0.19 (0.02)              & 2.17 (1.31)               \\
                              & (2-b)                      & 0.17 (0.02)                & 2.66 (0.10)              & 12.46 (0.36)              \\
                              & (3-c), $\alpha=0.2$        & 0.38 (0.06)                & 2.63 (0.11)              & 12.47 (0.46)              \\
                              & (3-c), $\alpha=0.4$        & 0.95 (0.27)                & 2.66 (0.09)              & 12.61 (0.59)              \\
                              & (4-c), $\alpha=0.2$        & 0.25 (0.03)                & 2.65 (0.12)              & 12.63 (0.39)              \\
                              & (4-c), $\alpha=0.4$        & 0.38 (0.08)                & 2.61 (0.10)              & 12.62 (1.70)              \\
                              \midrule
\multirow{4}{*}{Scenario II}  & (1-a)                      & 0.16 (0.12)                & 2.76 (0.10)              & 5.42 (2.37)               \\
                              & (1-b)                      & 0.09 (0.02)                & 2.84 (0.14)              & 15.82 (0.88)              \\
                              & (2-a)                      & 0.29 (0.38)                & 2.88 (0.13)              & 6.17 (3.14)               \\
                              & (2-b)                      & 0.14 (0.04)                & 2.79 (0.09)              & 15.46 (0.47)              \\
                              \midrule
\multirow{6}{*}{Scenario III} & (1-a)                      & 0.06 (0.02)                & 2.56 (0.20)              & 3.38 (2.79)               \\
                              & (1-b)                      & 0.07 (0.08)                & 2.52 (0.06)              & 11.79 (0.38)              \\
                              & (2-a)                      & 0.54 (0.67)                & 2.53 (0.09)              & 3.02 (2.39)               \\
                              & (2-b)                      & 0.56 (0.54)                & 2.61 (0.07)              & 12.03 (0.33)              \\
                              & (3-a)                      & 0.46 (0.46)                & 2.63 (0.13)              & 3.48 (2.75)               \\
                              & (3-b)                      & 0.36 (0.22)                & 2.71 (0.22)              & 12.41 (0.59)    \\
                          \bottomrule
\end{tabular}
\caption{Running time comparison (mean and standard deviation (in parenthesis), seconds) across Scenarios I, II, and III with $(n,p)=(400,20)$.}
\label{tab:time_comparison}
\end{table}

\section{Real Data Applications}\label{sec:real_datasets}
In this section, we test the proposed Fd-SDR method on four real-world datasets: global human mortality data, Washington D.C. bike rental count data, gene expression data from eleven types of carcinoma, and breast cancer survival data. 

\subsection{Global Human Mortality Data}

In this section, we demonstrate our method by applying it to the global human mortality dataset, which explores what influences human lifespans. This dataset captures age-at-death distributions, also referred to as mortality distributions, which are influenced by various factors such as economic conditions, healthcare systems, and social and environmental determinants. Insights from this data are crucial for developing policies aimed at improving health outcomes and extending life expectancy.

The dataset was sourced from the UN World Population Prospects 2019 Database \$available at \href{https://population.un.org/wpp/Download}{https://population.un.org/wpp/Download}). Nine predictors were used in the experiments: 
\begin{itemize}
    \item Population Density: population per square kilometer.  
    \item Sex Ratio: number of males per 100 females in the population.
    \item Mean Childbearing Age: average age of mothers at childbirth.  
    \item Gross Domestic Product (GDP) per Capita
    \item Gross Value Added (GVA) by Agriculture: percentage of GVA contributed by agriculture, hunting, forestry, and fishing.  
    \item Consumer Price Index (CPI): relative to the base year 2010.
    \item Unemployment Rate.  
    \item Expenditure on Health: as a percentage of GDP. 
    \item Arable Land: percentage of total land area.  
These predictors were gathered from United Nations databases.
\end{itemize}
This dataset offers a comprehensive view of demographic metrics, particularly mortality patterns across countries and regions. By examining age-specific mortality distributions, this data helps identify trends and disparities, offering valuable insights for global health policy.

Following the settings described in \citep{zhang2023dimension}, we analyzed life tables for each country and age group. The tables record the number of deaths $d(x,n)$ aggregated every five years. These data were treated as histograms, with bin widths of five years. For 162 countries in 2015, we smoothed these histograms using the `frechet' R package to obtain smoothed probability density functions. We then computed Wasserstein distances between the smoothed distributions. Additionally, we used a Gaussian kernel $\kappa(y,y') = \exp(-\gamma d_W^2(y,y'))$, with details provided in Section~\ref{sec:gen_fsdr}. We estimate the Fr\'echet central subspace by Fd-SDR, FOPG and GWIRE with $d=3$ and subsequently fit the projected data, $\hat{\bbeta}^T \bX$, using linear, polynomial and Fr\'echet regression \citep{petersen2019Frechet}. The adjusted R-squared values are provided in Table~\ref{tab:Rsqmort}. The R-squared of Fd-SDR and GWIRE methods are close and both of them better than FOPG method (Table~\ref{tab:Rsqmort}). Figure \ref{fig:morbidity_three_images} (a) shows the mortality rate over ages versus the 1st sufficient predictor. (b) and (c) show that the mean and standard deviation of age at death of each country versus the 1st and 2nd sufficient predictors. 

\begin{table}[ht]
    \centering
    \renewcommand{\arraystretch}{1.3} % Adjust row height for better readability
    \setlength{\tabcolsep}{12pt} % Adjust column spacing
    \begin{tabular}{l|c|c|c}
        \hline
        \textbf{Metric} & \textbf{Fd-SDR} & \textbf{FOPG} & \textbf{GWIRE} \\ 
        \hline\hline
        Adjusted R-squared  & 0.542 & 0.415 & 0.561 \\ 
        Fréchet R-squared   & 0.544 & 0.386 & 0.562 \\ 
        Poly R-squared      & 0.929 & 0.932 & 0.933 \\ 
        \hline
    \end{tabular}
    \caption{Performance metrics for the global mortality data}
    \label{tab:Rsqmort}
\end{table}

\begin{figure}[ht]
    \centering
    \begin{subfigure}[b]{0.3\textwidth}
        \centering
        \includegraphics[width=\textwidth]{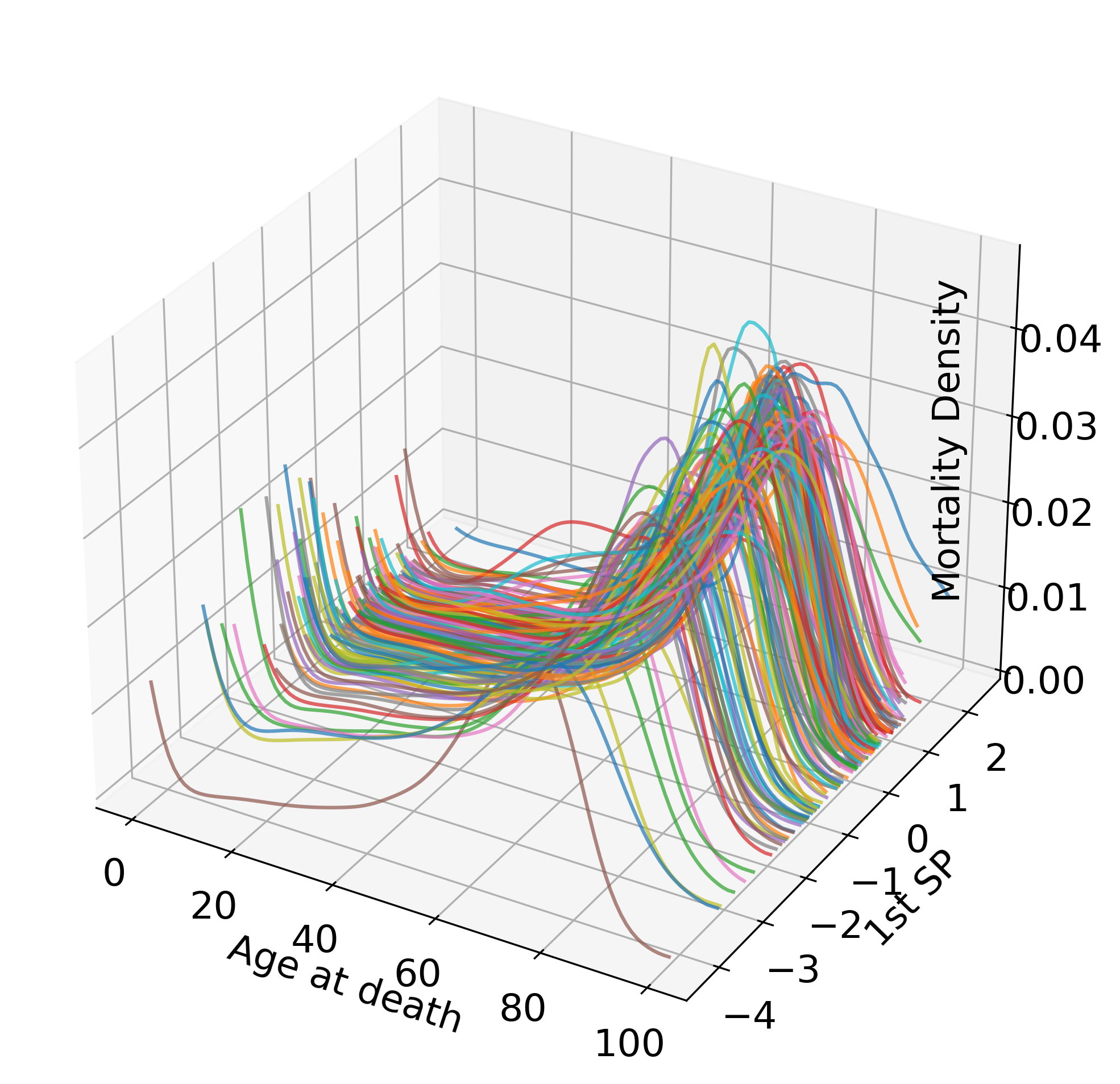}
        \caption{1st sufficient predictor}
           \end{subfigure}
    \hfill
    \begin{subfigure}[b]{0.33\textwidth}
        \centering
        \includegraphics[width=\textwidth]{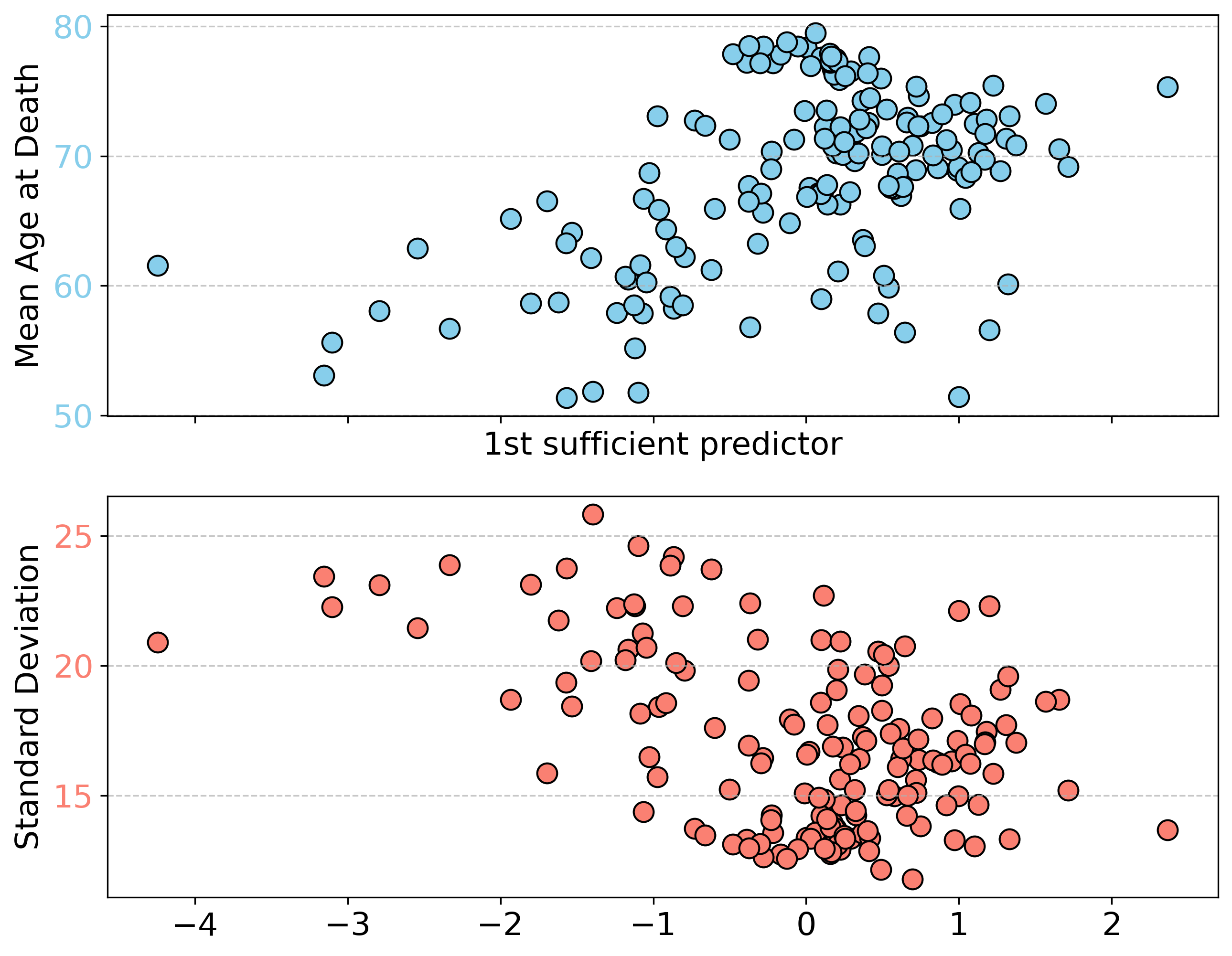}
        \caption{}
          \end{subfigure}
    \hfill
    \begin{subfigure}[b]{0.33\textwidth}
        \centering
        \includegraphics[width=\textwidth]{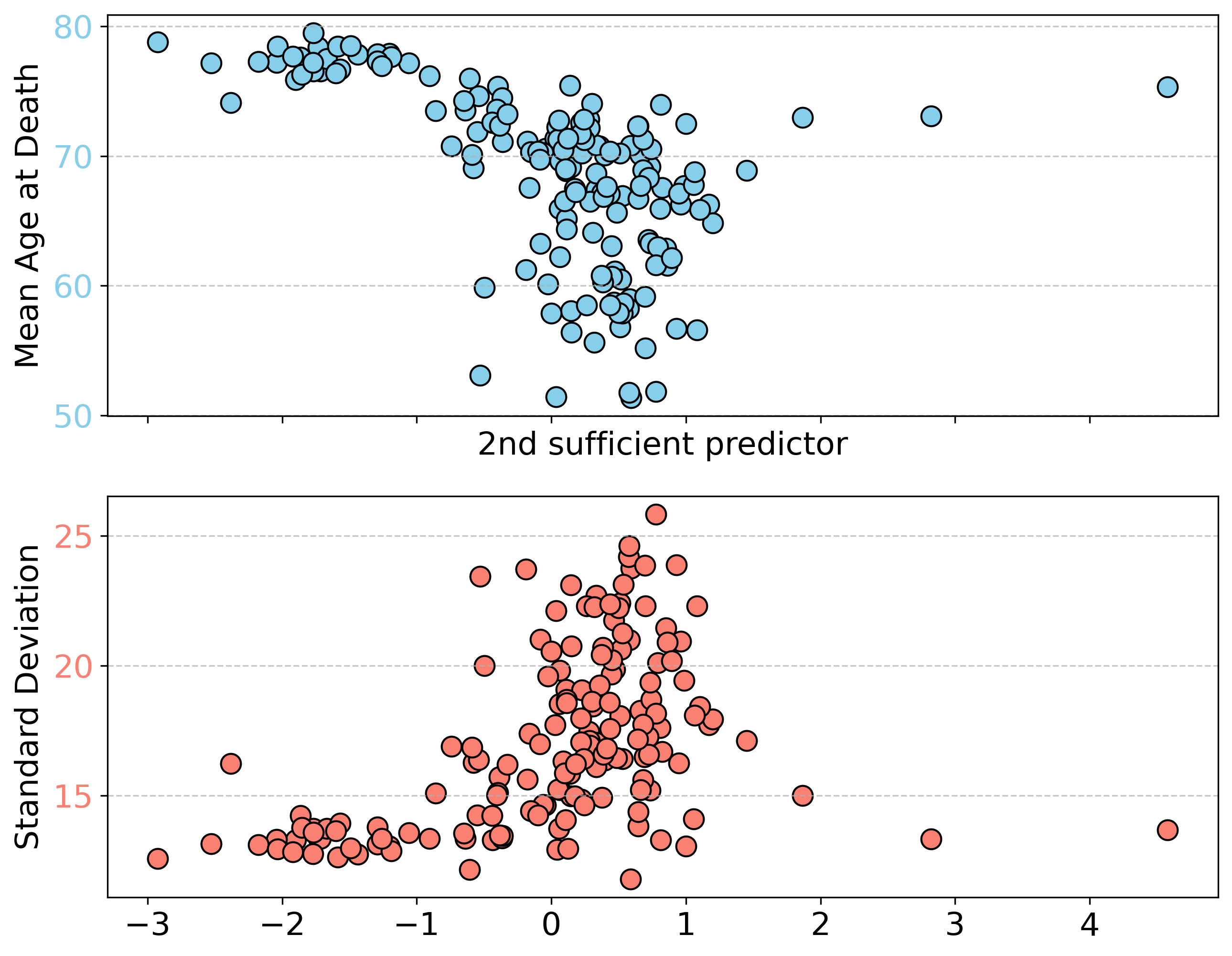}
        \caption{}
           \end{subfigure}
    \caption{(a) Global mortality rate by age (in years) plotted against the first sufficient predictor; (b) Mean mortality distribution versus the first two sufficient predictors; (c) Standard deviation of mortality distribution versus the first two sufficient predictors.}
    \label{fig:morbidity_three_images}
\end{figure}

\subsection{Bike Rental Data}
In this section, we apply our proposed method to bike rental data \citep{fanaee2014event}. The bike rental data from the UCI Machine Learning Repository is a widely-used dataset that captures bike-sharing usage patterns in Washington, D.C. The data spans two years (2011–2012) and includes hourly and daily counts of bike rentals, along with information on weather conditions, seasonal effects, holidays, and time-based features such as hour and weekday. This dataset is ideal for exploring temporal patterns in bike rentals, assessing the impact of external factors like weather and holidays on usage, and developing predictive models for demand forecasting.
We follow the settings from \citep{weng2023sparse} and choose 6 predictors similar to \citep{weng2023sparse}:  
1) Holiday: Indicator of a public holiday celebrated in Washington D.C.
2) Workingday: Indicator of neither the weekend nor a holiday
3) Temp: Daily mean temperature
4) Atemp: Feeling temperature
5) BW: Indicator of bad weather 
6) RBW: Indicator of really bad weather.

In this dataset, we set $d=2$ to estimate the Fr\'echet central subspace with the Gaussian kernel $\kappa(y,y') = \exp(-\gamma d_W^2(y,y'))$. We fit the estimated projected data, $\hat{\bbeta}^T \bX$, using linear, polynomial and Fr\'echet regression \citep{petersen2019Frechet}. The corresponding values of R squares are presented in Table~\ref{tab:Rsqbike} and the coefficients of the estimated subspace are given in Table~\ref{tab:bike_beta}. Among the three methods, Fd-SDR has the highest R-squared values in all regression models. 
In addition, we display the bike rental count data over time (by hour) against the first and second sufficient predictors in Figures~\ref{fig:Fd-SDR} and \ref{fig:maxdays}, respectively. Figure~\ref{fig:Fd-SDR} (a) illustrates all curves for both working and nonworking days, while Figures~\ref{fig:Fd-SDR}(b) and (c) separately depict the curves for working and non-working days, respectively. Based on the curve patterns shown in Figure~\ref{fig:Fd-SDR}, it is evident that bike rentals are common in the morning before 8 AM and in the afternoon around 6 PM. On non-working days, people predominantly rent bikes before 12 PM. We infer that on working days, most bike rentals are for commuting, whereas on non-working days, people complete their outdoor activities in the morning.

Furthermore, Table~\ref{tab:bike_beta} indicates that Fd-SDR identifies two types of factors influencing bike rentals. The first type consists of individuals who commute using rental bikes ($\hat{\bbeta_1}$ shows a high loading on the predictor ``Working'') and who consider the perceived temperature (``Atemp'' coefficient is larger than other predictors in magnitude). The second type comprises those who rent bikes mainly for entertainment, where temperature and weather are the primary factors (see coefficients for $\hat{\bbeta_2}$). The results obtained by Fd-SDR align with common observations.

% Figure \ref{fig:maxdays} illustrates the maximum counts of working and non-working days plotted against the second sufficient predictor, which has a high loading on the predictor Atemp. This indicates that people are more likely to use bikes when the temperature is moderate, neither too cold nor too hot.
 
\begin{table}[!ht]
    \centering
    \renewcommand{\arraystretch}{1.3} % Adjust row height for better readability
    \setlength{\tabcolsep}{12pt} % Adjust column spacing
    \begin{tabular}{l|c|c|c}
        \hline
        \textbf{Metric} & \textbf{Fd-SDR} & \textbf{FOPG} & \textbf{GWIRE} \\ 
        \hline\hline
        Adjusted R-squared & 0.519 & 0.361 & 0.504 \\ 
        Fréchet R-squared & 0.367 & 0.344 & 0.365 \\ 
        Polynomial R-squared & 0.783 & 0.704 & 0.769 \\ 
        \hline
    \end{tabular}
    \caption{Performance metrics for the bike rental data}
    \label{tab:Rsqbike}
\end{table}

\begin{table}[!ht]
    \centering
    \renewcommand{\arraystretch}{1.3} % Adjust row height for better readability
    \setlength{\tabcolsep}{12pt} % Adjust column spacing
    \begin{tabular}{l|c|c}
        \hline
        \textbf{Feature} & $\hat{\bbeta}_1$ & $\hat{\bbeta}_2$ \\
        \hline\hline
        Holiday & -0.01 & -0.07 \\
        Working & -0.99 & -0.15 \\
        Temp    & 0.05  & -0.79 \\
        Atemp   & -0.13 & 1.7 \\
        BW      & -0.02 & 0.16 \\
        RBW     & 0.03  & -0.2 \\
        \hline
    \end{tabular}
    \caption{Coefficients of the two directions $\hat{\bbeta}_1$ and $\hat{\bbeta}_2$ for the bike rental data, which are obtained by Fd-SDR.}
    \label{tab:bike_beta}
\end{table}

\begin{figure}[ht]
    \centering
    \begin{subfigure}[b]{0.32\textwidth}
        \centering
        \includegraphics[width=\textwidth]{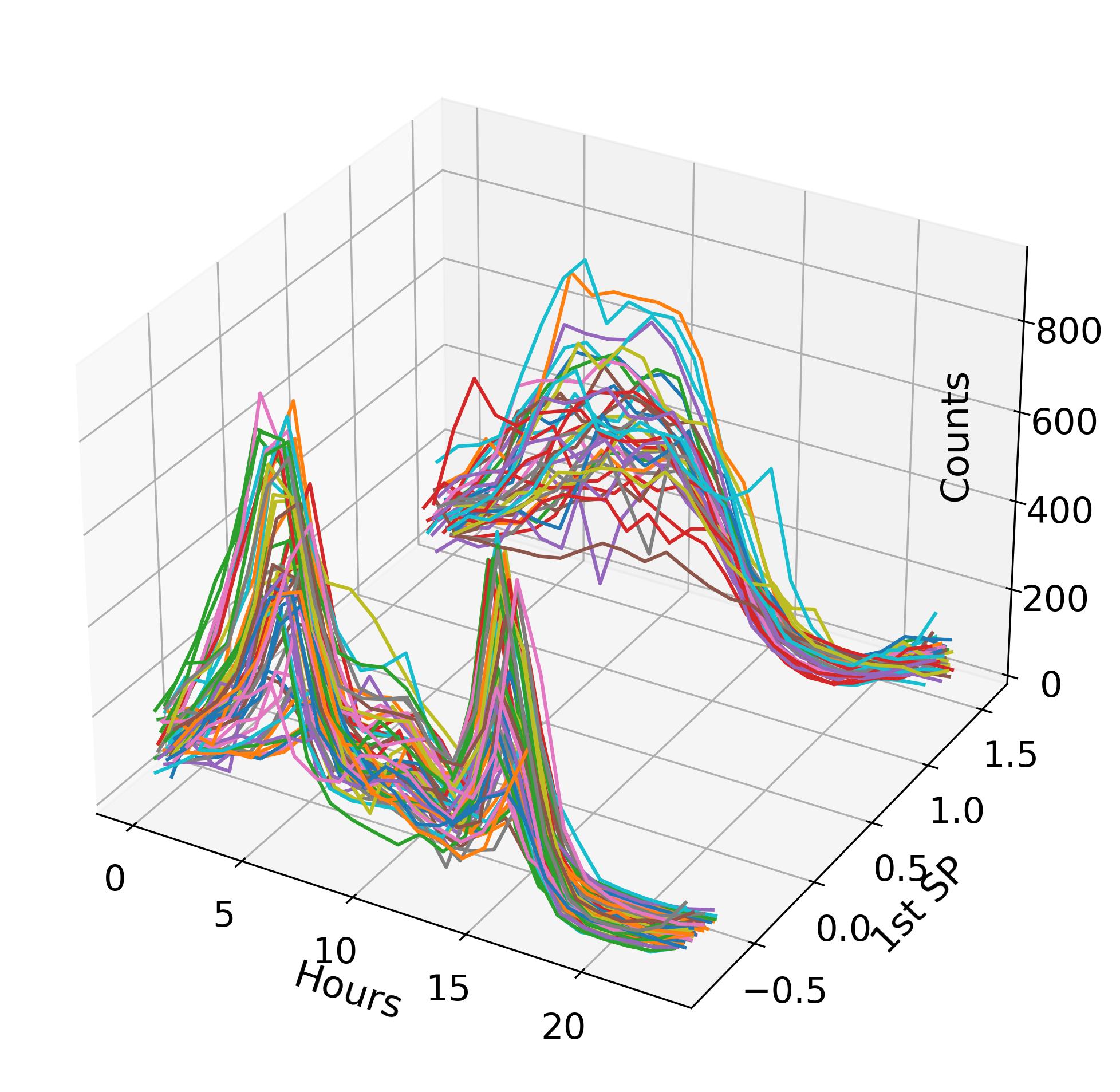}
        \caption{First sufficient predictor: working days (front) and non-working days (back).}
          \end{subfigure}
    % \hfill
    \begin{subfigure}[b]{0.32\textwidth}
        \centering
        \includegraphics[width=\textwidth]{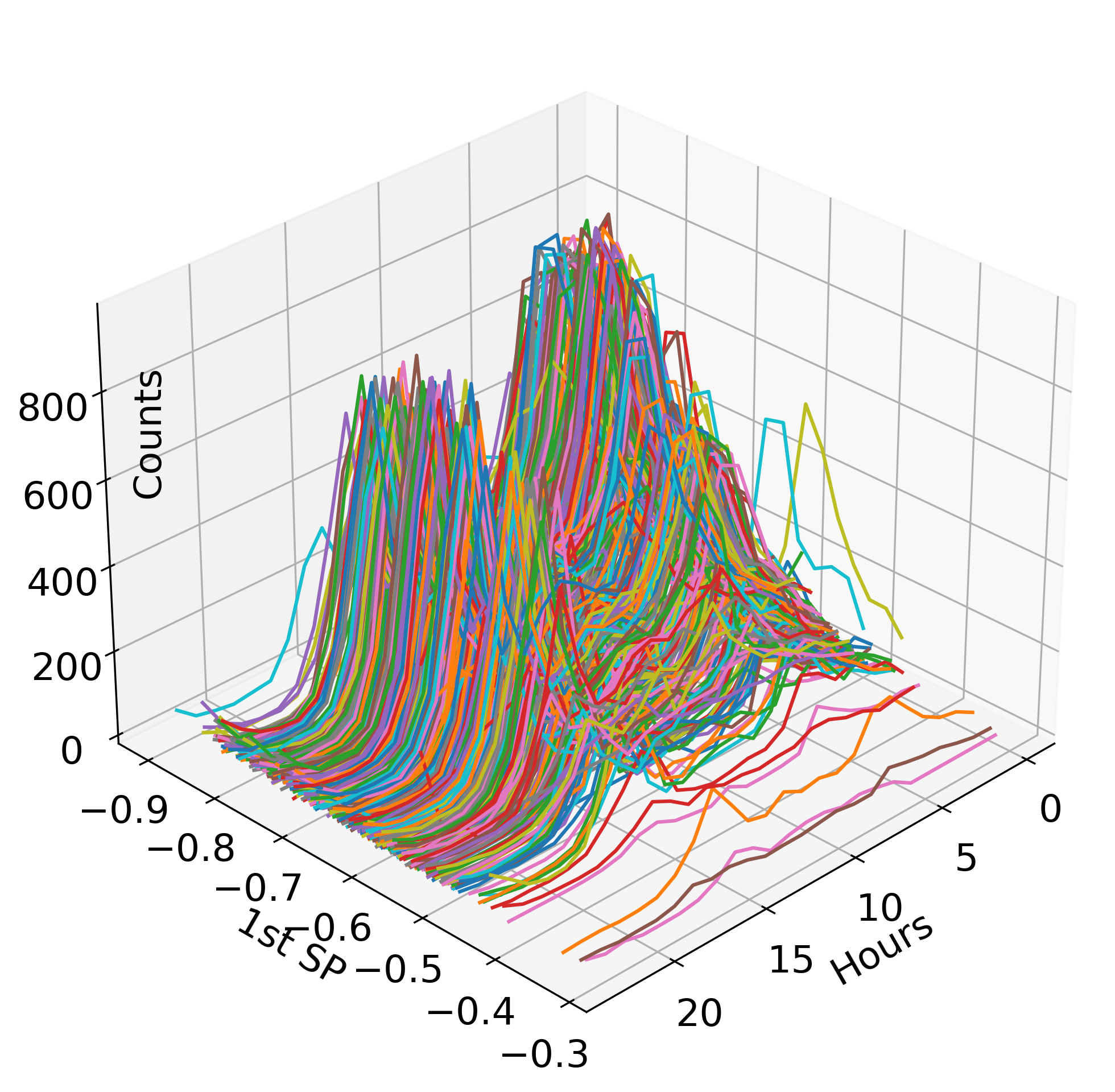}
        \caption{First sufficient predictor for working days.}
          \end{subfigure}
    % \hfill
    \begin{subfigure}[b]{0.32\textwidth}
        \centering
        \includegraphics[width=\textwidth]{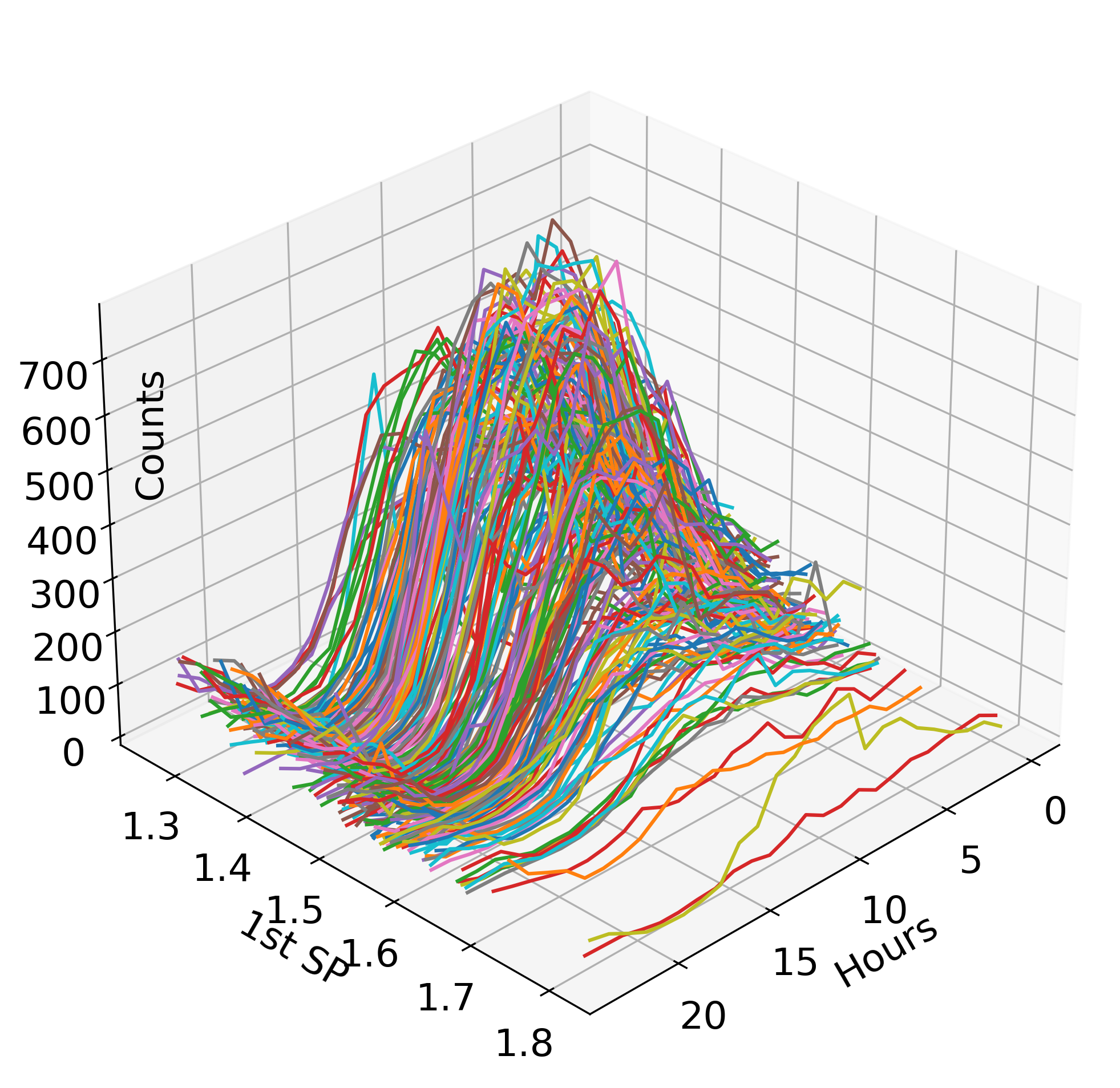}
        \caption{First sufficient predictor for non-working days.}
          \end{subfigure}
    \caption{The bike rental count data over time (by hour) projected onto the first sufficient predictor in (a) and onto the first sufficient predictors in (b) and (c).}
    \label{fig:Fd-SDR}
\end{figure}

\begin{figure}[ht]
    \centering
    \includegraphics[width=0.45\textwidth]{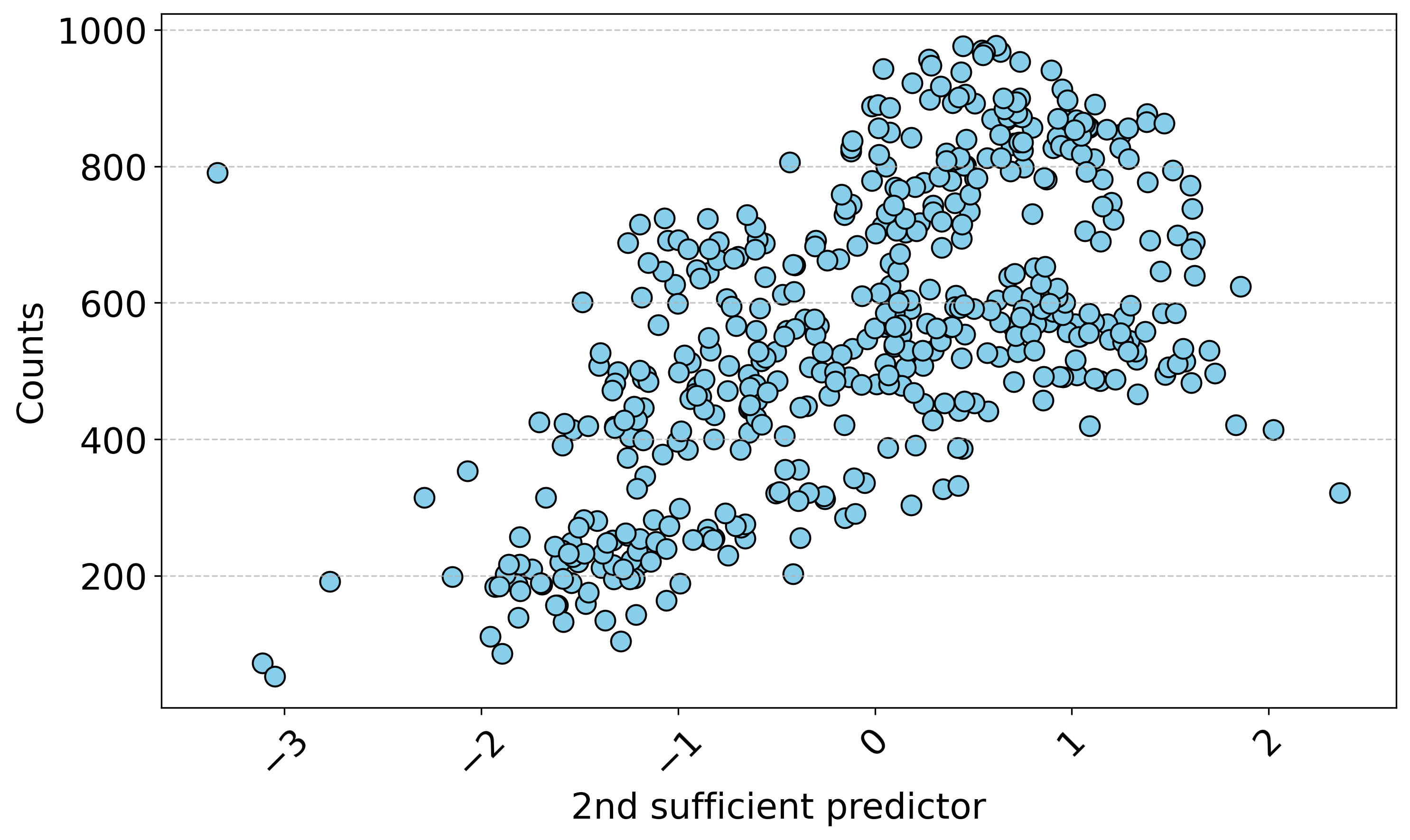}
    \includegraphics[width=0.45\textwidth]{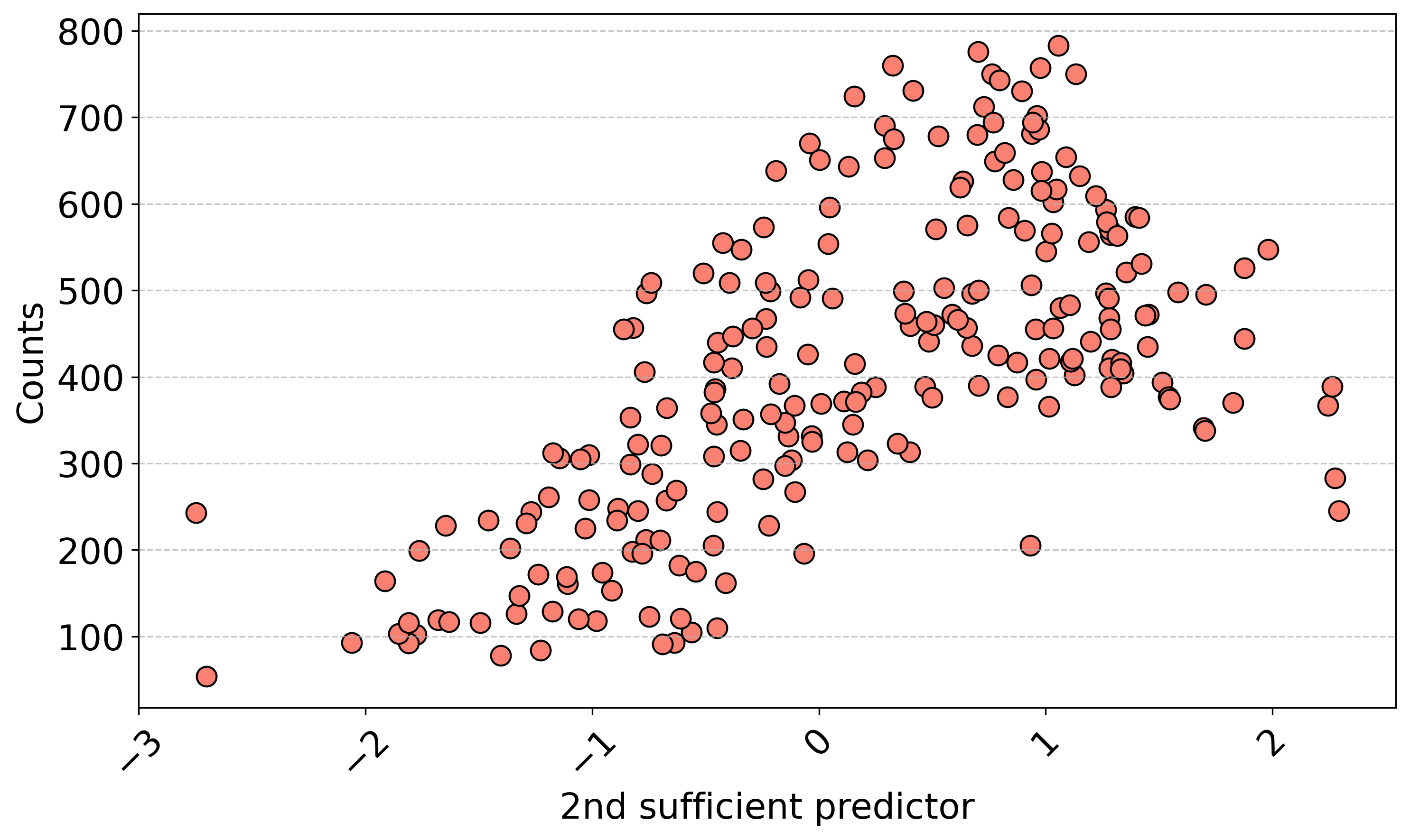}
    % \begin{subfigure}[b]{0.45\textwidth}
    %     \centering
    %     \includegraphics[width=\textwidth]{bike4.png}
    %     \caption{Working days}
    %      \end{subfigure}
    % % \hfill
    % \begin{subfigure}[b]{0.45\textwidth}
    %     \centering
    %     \includegraphics[width=\textwidth]{bike5.png}
    %     \caption{Nonworking days}
    %    \end{subfigure}
    \caption{Plots of the maximum counts of bike rental distributions versus the second sufficient predictor on working days (left) and non-working days (right).}
    \label{fig:maxdays}
\end{figure}

\subsection{Ultrahigh-Dimensional Analysis: Carcinomas Data} \label{applications3}

Carcinoma, a type of cancer originating in epithelial cells that form the skin or line internal organs, requires precise classification based on primary anatomical sites (e.g., prostate, liver) to guide optimal treatment strategies \citep{su2001molecular}. This study examines the carcinomas dataset (U95a GeneChip) from \citep{su2001molecular}, comprising $n = 174$ samples spanning 11 carcinoma types: prostate, bladder/ureter, breast, colorectal, gastroesophageal, kidney, liver, ovary, pancreas, lung adenocarcinoma, and lung squamous cell carcinoma. For clarity, we denote these types as classes 0 through 10 in our experiments. The respective sample sizes are 26, 8, 26, 23, 12, 11, 7, 27, 6, 14, and 14. Collectively, these carcinoma types account for approximately 70\% of cancer-related deaths in the United States \citep{su2001molecular}.
Each sample provides gene expression levels for $p = 9183$ predictors, which have been preprocessed as described by \citep{su2001molecular}. The primary goal of this analysis is to classify carcinoma types based on gene expression profiles and identify genes significantly associated with each carcinoma category.

We applied the proposed Fd-SDR method to this dataset, treating carcinoma types as the response variable. To reduce dimensionality, we employed the Ball Correlation Sure Independence Screening (BCor-SIS) method \citep{pan2019generic}, which reduced the number of predictors to $173$. Fd-SDR, FOPG, and GWIRE with a Gaussian kernel are then used to estimate the Fr\'echet central subspace with dimension $d=2$. Logistic regression was then performed using $\hat{\bbeta}^T\bX $ as the predictors and $\bY$ as the responses. Model performance was evaluated using the area under the curve (AUC) metric. The mean and standard deviation of the AUC scores, obtained through 5-fold cross-validation, are presented in Table~\ref{tab:AUCs_carcinom}. As shown, Fd-SDR achieved the highest AUC, outperforming both FOPG and GWIRE.

We also visualize the projected data $(\hat{\bbeta}^{\text{Fd-SDR}})^T\bX$ in Figure \ref{fig:Carcinoma_plot}. The plot reveals clear separation among most carcinoma types, with some overlap observed between colorectal, kidney, and ovarian carcinomas (classes 4, 6, and 8). This highlights the effectiveness of our proposed method. The observed overlaps are likely due to shared genomic mutations and similar metastatic patterns, as noted in \citep{kan2010diverse,kir2010clinicopathologic}.

\begin{table}[ht]
    \centering
    \renewcommand{\arraystretch}{1.3} % Adjust row height for better readability
    \setlength{\tabcolsep}{12pt} % Adjust column spacing
    \begin{tabular}{l|c|c|c}
        \hline
        \textbf{Method} & \textbf{Fd-SDR} & \textbf{FOPG} & \textbf{GWIRE} \\ 
        \hline\hline
        AUC & 0.796 (0.180) & 0.515 (0.155) & 0.714 (0.213) \\ 
        \hline
    \end{tabular}
    \caption{AUC Mean with standard deviation (in parentheses) for logistic regression.}
    \label{tab:AUCs_carcinom}
\end{table}

\begin{figure}[ht]
    \centering
    \includegraphics[width=0.8\textwidth]{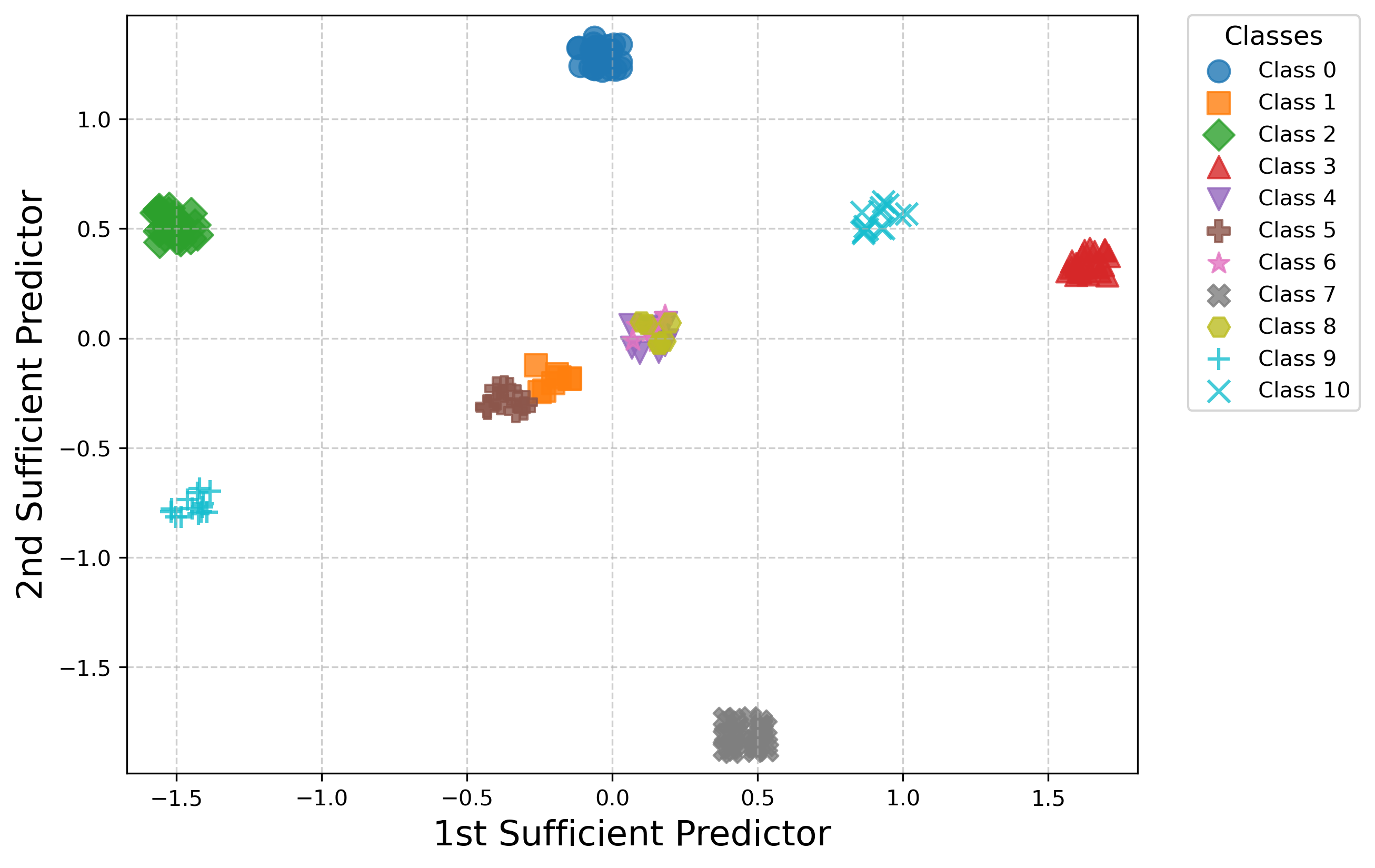} % Change to .pdf for vector graphics
    \caption{Scatter plot of the first two sufficient predictors for the carcinoma data, with labels ranging from 0 to 10.}
    \label{fig:Carcinoma_plot}
\end{figure}

\subsection{Ultrahigh-Dimensional Analysis: TCGA Breast Cancer Survival Data}\label{applications4}

The last experiment to test the proposed method is the Breast Cancer dataset from The Cancer Genome Atlas (TCGA), utilizing three sources of data: gene expression RNAseq (IlluminaHiSeq percentile), curated survival endpoints, and clinical data obtained from the Xena platform \citep{goldman2020visualizing}. The goal of this experiment to investigate the relationship between overall survival time and gene expressions.

The gene expression dataset consists of $1,218$ samples and $20,531$ genes. The curated survival data includes four types of survival endpoints for each TCGA cancer type: Overall Survival (OS), Progression-Free Interval (PFI), Disease-Free Interval (DFI), and Disease-Specific Survival (DSS). Additionally, the clinical dataset provides information on each patient's age at the time of initial pathological diagnosis.

We divided all samples into 34 age-based groups and calculated the survival time distributions for each group. The average gene expression within each group was used as predictors. Since the dimensionality is much larger than the number of samples, we applied the BCor-SIS method \citep{pan2019generic} to reduce the dimensionality. We then applied the Fd-SDR, FOPG, and GWIRE methods to estimate the Fr\'echet central subspace with a dimension of $ d = 3 $. The estimated projected data, $\hat{\bbeta}^T\bX$, was fitted using linear, polynomial and Fr\'echet regression. The R-squared values are presented in Table \ref{tab:Rsqbc}, demonstrating that Fd-SDR has the highest values and outperforms the other two methods in both linear and Fr\'echet regression models. Since all three estimators achieve exceptional performance with R-squared values of $0.99$ for polynomial regression, these results are omitted from Table \ref{tab:Rsqbc}. 

\begin{table}[ht]
    \centering
    \renewcommand{\arraystretch}{1.3} % Adjust row height for better readability
    \setlength{\tabcolsep}{12pt} % Adjust column spacing
    \begin{tabular}{l|c|c|c}
        \hline
        \textbf{Metric} & \textbf{Fd-SDR} & \textbf{FOPG} & \textbf{GWIRE} \\ 
        \hline\hline
        Adjusted R-squared  & 0.661 & 0.645 & 0.648 \\ 
        Fréchet R-squared   & 0.702 & 0.688 & 0.691 \\ 
        \hline
    \end{tabular}
    \caption{R-squared metrics for the breast cancer survival data.}
    \label{tab:Rsqbc}
\end{table}

%\begin{figure}[ht]
%    \centering
%    \begin{subfigure}[b]{0.45\textwidth}
%        \centering
%        \includegraphics[width=\textwidth]{survival1.png}
        %\caption{Working days}
%        \label{fig:survival1}
%    \end{subfigure}
    % \hfill
%    \begin{subfigure}[b]{0.45\textwidth}
    
%        \centering
%        \includegraphics[width=\textwidth]{survival2.png}
        %\caption{Nonworking days}
%        \label{fig:survival2}
%    \end{subfigure}
%    \caption{The overall survival (OS) time by months versus (a) the age (left) and (b) the 1st sufficient predictor (right).}
%    \label{fig:survival}
%\end{figure}

\section{Discussion}

In classical regression, SDR is a powerful tool for exploratory data analysis, regression diagnostics, and mitigating the curse of dimensionality. By addressing challenges such as collinearity among predictors, heteroscedasticity in the response, and identifying the most critical linear combinations of predictors, SDR facilitates more efficient and interpretable data analysis. By examining scatter plots of the response against the first two sufficient predictors derived via SDR, one can visualize and explore the general shape of the regression surface without resorting to complex models. These visualizations are especially valuable for uncovering and interpreting relationships in high-dimensional settings.

In our proposed method, the projection of mortality density onto the first sufficient predictor within the CMS effectively captures key patterns across various applications, including mortality distributions between countries, bike rental count distributions, carcinoma clusters in two-dimensional projection plots, and breast cancer survival data. For example, in the case of mortality data, the lower end of the sufficient predictor reveals a shift toward higher longevity, while the upper end highlights lower longevity with a marked increase near age 0, reflecting infant mortality. These results highlight the capability of CMS to condense complex, multivariate relationships into an interpretable low-dimensional representation.

A key innovation of our approach is the integration of kernel distance covariance into the SDR framework, which enhances the detection of complex and non-linear functional relationships. Traditional SDR methods often emphasize linear dependencies, which can limit their ability to uncover more intricate associations. By incorporating kernel distance covariance, our method expands the scope of SDR, allowing it to capture a broader range of associations and achieve more flexible and precise dimension reduction. This advancement broadens the applicability of SDR to accommodate response variables valued in metric spaces, with the CMS framework providing a solid foundation for identifying sufficient predictors. Furthermore, our methodology applies to any separable and complete metric space of negative type, greatly expanding the potential applications of SDR and offering a robust framework for Fr\'echet regression in diverse data settings.

In conclusion, the integration of kernel distance covariance within the CMS framework greatly enhances the applicability and versatility of SDR for analyzing metric-space-valued data. This advancement facilitates the discovery of complex relationships in high-dimensional settings, representing a significant step forward in the theory and practice of sufficient dimension reduction.

\section*{Funding}
We gratefully acknowledge the support of the \textit{National Science Foundation} through grants (DMS-1924792, DMS-2318925 and CNS-1818500).
\bibliography{FdSDR}

\end{document}